\documentclass[11pt]{llncs}

\usepackage{fullpage}
\usepackage{graphicx, color}
\usepackage{latexsym,amsfonts,enumerate}
\usepackage{amsmath}
\usepackage{amssymb}
\usepackage{dsfont}
\usepackage{tikz} 
\usetikzlibrary{arrows,decorations.markings,patterns,calc, positioning, backgrounds} 

\newtheorem{new-claim}{Claim}
\newcommand{\vote}{\mathsf{vote}}
\newcommand{\capac}{\mathsf{cap}}
\newcommand{\Nbr}{\mathsf{nbr}}

\begin{document}
\pagestyle{plain}
\title{Popular matchings with two-sided preferences and one-sided ties\thanks{A preliminary version of this paper has appeared at ICALP~2015.}}
\author{\'{A}gnes Cseh\inst{1}\thanks{Work done while visiting TIFR, supported by the Deutsche Telekom Stiftung.}
\and Chien-Chung Huang\inst{2}\and Telikepalli Kavitha\inst{3}}
\institute{Reykjavik University, Iceland. \email{cseh@ru.is} \and Chalmers University, Sweden. \email{huangch@chalmers.se} \and Tata Institute of Fundamental Research, India. \email{kavitha@tcs.tifr.res.in}}
\maketitle

\begin{abstract}
We are given a bipartite graph $G = (A \cup B, E)$ where each vertex has a preference list
ranking its neighbors: in particular, every $a \in A$ ranks its neighbors in a strict order 
of preference, whereas the preference lists of $b \in B$ may contain ties. A matching $M$ is 
\emph{popular} if there is no matching $M'$ such that the number of vertices that prefer $M'$ 
to $M$ exceeds the number of vertices that prefer $M$ to~$M'$. We show that the problem of deciding whether 
$G$ admits a popular matching or not is $\mathsf{NP}$-hard. This is the case even when every 
$b \in B$ either has a strict preference list or puts all its neighbors into a single tie. 
In contrast, we show that the problem becomes polynomially solvable in the case when each 
$b \in B$ puts all its neighbors into a single tie. That is, all neighbors of $b$ are tied 
in $b$'s list 
and $b$ desires to be matched to any of them. Our main result is an $O(n^2)$ algorithm 
(where $n = |A \cup B|$) for the popular matching problem in this model. Note that this model 
is quite different from the model where vertices in $B$ have no preferences and do {\em not} 
care whether they are matched or not.
\end{abstract}

\section{Introduction}
\label{sec:intro}
We are given a bipartite graph $G = (A \cup B, E)$ where the vertices in $A$ are called 
applicants and the vertices in $B$ are called posts, and each vertex has a preference list 
ranking its neighbors in an order of preference. Here we assume that vertices in $A$ have 
strict preferences while vertices in $B$ are allowed to have ties in their preference lists. 
Thus each applicant ranks all posts that she finds interesting in a strict order of preference,
while each post need not come up with a total order on all interested applicants -- here 
applicants may get grouped together in terms of their suitability, thus equally competent 
applicants are tied together at the same rank. 

Our goal is to compute a {\em popular} matching in~$G$. The definition of popularity uses the 
notion of each vertex casting a ``vote'' for one matching versus another. A vertex $v$ 
{\em prefers} matching $M$ to matching $M'$ if either $v$ is unmatched in $M'$ and matched in 
$M$ or $v$ is matched in both matchings and $M(v)$ ($v$'s partner in $M$) is ranked better 
than $M'(v)$ in $v$'s preference list. In an election between matchings $M$ and $M'$, each 
vertex $v$ votes for the matching that it prefers or it abstains from voting if $M$ and $M'$ 
are equally preferable to~$v$. Let $\phi(M,M')$ be the number of vertices that vote for $M$ in 
an election between $M$ and~$M'$.

\begin{definition}
A matching $M$ is \emph{popular} if $\phi(M,M') \ge \phi(M',M)$ for every matching~$M'$.
\end{definition}

If $\phi(M',M) > \phi(M,M')$, then we say $M'$ is {\em more popular} than $M$ and denote it by 
$M' \succ M$; else~$M \succeq M'$. Observe that popular matchings need not always exist. 
Consider an instance where $A = \{a_1,a_2,a_3\}$ and $B = \{b_1,b_2,b_3\}$ and for $i = 1,2,3$, 
each $a_i$ has the same preference list which is $b_1$ followed by $b_2$ followed by $b_3$ 
while each $b_i$ ranks $a_1,a_2,a_3$ the same, i.e.\ $a_1,a_2,a_3$ are tied together in $b_i$'s 
preference list (see bottom left instance in Fig.~\ref{fig:AppA}). It is easy to see that for any matching $M$ 
here, there is another matching $M'$ such that $M' \succ M$, thus this instance admits no 
popular matching.

The popular matching problem is to determine if a given instance $G = (A\cup B,E)$ admits a 
popular matching or not, and if so, to compute one. This problem has been studied in the 
following two models.
\begin{itemize}
\item {\em 1-sided model:} here it is only vertices in $A$ that have preferences and cast 
votes; vertices in $B$ are objects with no preferences or votes.
\item {\em 2-sided model:} vertices on both sides have preferences and cast votes.
\end{itemize}

Popular matchings need not always exist in the 1-sided model and the problem of whether a 
given instance admits one or not can be solved efficiently using the characterization and 
algorithm from~\cite{AIKM05}. In the 2-sided model when all preference lists are strict, it 
can be shown that any stable matching is popular; thus a popular matching can be 
found in linear time using the Gale-Shapley algorithm. However when ties are allowed in 
preference lists on both sides, Bir\'o, Irving, and Manlove~\cite{BIM09} showed that the 
popular matching problem is $\mathsf{NP}$-complete. In this paper we focus on the following 
variant:
\begin{enumerate} [$\ast$]
\item it is only vertices in $A$ that have preference lists ranking their neighbors, 
however vertices on {\em both} sides cast votes.
\end{enumerate}

That is, vertices in $B$ have no preference lists ranking their neighbors -- however each 
$b \in B$ desires to be matched to any of its neighbors. Thus in an election between two 
matchings, $b$ abstains from voting if it is matched in both or unmatched in both, else it 
votes for the matching where it is matched. An intuitive understanding of such an instance 
is that $A$ is a set of applicants and $B$ is a set of tasks -- while each applicant has a 
preference list over the tasks that she is interested in, each task just cares to be assigned 
to anyone who is interested in performing it. 

For instance, each task is a building that seeks to have some guard assigned to it and it has no 
preferences over the identity of the guard. Another application is in the {\em many-to-one} 
popular matching problem  where each $b \in B$ also has a capacity $\capac(b)$ associated with 
it. Here we seek a popular matching that can match every $b \in B$ to up to $\capac(b)$-many 
neighbors and we need to devise natural and succinct rules to decide when $b$ prefers one subset 
over another. A possible model is to say that $b$ just cares to have enough partners in the 
matching and it does not care about the identities of these partners. That is, we say $b$ 
prefers $M_1$ to $M_2$ if $\capac(b) \ge |M_1(b)| > |M_2(b)|$, where $|M(b)|$ is the number of 
partners of $b$ in~$M$. 

In this paper we deal with the one-to-one setting of the above problem, i.e.\ $\capac(b) = 1$
for all $b \in B$. We will see in 
Section~\ref{sec:prelims} that this problem is significantly different 
from the popular matching problem in the 1-sided model where vertices in $B$ do not cast votes.  
We show the following results here, complementing our polynomial time algorithm in 
Theorem~\ref{thm:popular} with our hardness result in Theorem~\ref{thm:nphard}.

\begin{theorem}
\label{thm:popular}
Let $G = (A\cup B,E)$ be a bipartite graph where each $a \in A$ has a strict preference list 
while each $b \in B$ puts all its neighbors into a single tie. The popular 
matching problem in $G$ can be solved in $O(n^2)$ time, where $|A\cup B| = n$.
\end{theorem}

\begin{theorem}
\label{thm:nphard}
Let $G = (A\cup B,E)$ be a bipartite graph where each $a \in A$ has a strict preference list 
while each $b \in B$ either has a strict preference list or puts all its neighbors into a single 
tie. The popular matching problem in $G$ is $\mathsf{NP}$-complete. 
\end{theorem}
 
Theorem~\ref{thm:nphard} follows from a simple reduction from the \textsc{(2,2)-e3-sat} problem. The \textsc{(2,2)-e3-sat} problem takes as its input a Boolean formula $\mathcal{I}$ in CNF, where each clause contains three literals and every variable appears exactly twice in unnegated form and exactly twice in negated form in the clauses. The problem is to determine if $\mathcal{I}$ 
is satisfiable or not. This problem is $\mathsf{NP}$-complete~\cite{BKS03} and our reduction shows that the following
version of the 2-sided popular matching problem in $G = (A\cup B,E)$ with 1-sided ties is $\mathsf{NP}$-complete: 
\begin{itemize}
\item every vertex in $A$ has a strict preference list of length 2 or 4;
\item every vertex in $B$ has either a strict preference list of length~2 or a single tie of length~2 or 3 as a preference list.
\end{itemize}
Note that our $\mathsf{NP}$-hardness reduction needs $B$ to have $\Omega(|B|)$ vertices with 
strict preference lists and $\Omega(|B|)$ vertices with single ties as their preference lists.

Our algorithm that proves Theorem~\ref{thm:popular} performs a partition of the set $B$ into three sets: 
the first set $X$ is a subset of top posts and, roughly speaking, the second set $Y$ consists of 
{\em mid-level} posts, while the third set $Z$ consists of {\em unwanted} posts 
(see Fig.~\ref{fig:next}). Applicants get divided into two sets: the set of those with one or 
more neighbors in the set $Z$ (call this set $\Nbr(Z)$) and the rest (this set is 
$A\setminus\Nbr(Z)$).  

Our algorithm performs the partition of $B$ into $X, Y$, and $Z$ over several iterations. 
Initially $X = F$, where $F$ is the set of top posts, $Y = B \setminus F$, and $Z = \emptyset$.
In each iteration, certain top posts get {\em demoted} from $X$ to $Y$ and certain non-top 
posts get demoted from $Y$ to~$Z$. With new posts entering $Z$, we also have applicants
moving from $A\setminus\Nbr(Z)$ to $\Nbr(Z)$. Using the partition $\langle X,Y,Z\rangle$ of 
$B$, we will build a graph $H$ where each applicant keeps at most two edges: either 
to its most preferred post in $X$ and also in $Y$ or to its most preferred post in $Z$ and 
also in~$Y$. Some dummy posts may be included in~$Y$.

We prove that $G$ admits a popular matching if and only if $H$ admits an $A$-complete matching,
i.e.\ one that matches all vertices in~$A$. 
We show that corresponding 
to any popular matching in $G$, there is a partition $\langle L_1,L_2,L_3\rangle$ of $B$ into 
{\em top posts}, {\em mid-level posts}, and {\em unwanted posts} such that $X \supseteq L_1$ 
and $Z \subseteq L_3$, where $\langle X,Y,Z\rangle$ is the partition computed by our algorithm.
This allows us to show that if $H$ does not admit an $A$-complete matching, then $G$ has no 
popular matching. In fact, not every popular matching in $G$ becomes an $A$-complete matching 
in $H$ (Section~\ref{sec:algo} has such an example). However it will be the case that if $G$ 
admits popular matchings, then at least one of them becomes an $A$-complete matching in~$H$.

\medskip

\noindent{\bf Background.\ \ } Popular matchings have been well-studied in the 1-sided model 
\cite{AIKM05,KN08,Mah06,MS06,MI09,Mes06} where only vertices of $A$ have preferences and cast votes. 
Abraham et al.~\cite{AIKM05} gave efficient algorithms to determine if a given instance admits a 
popular matching or not -- their algorithm also works when preference lists of vertices in $A$ admit ties. 
The notions of {\em least unpopular} matchings~\cite{McC06} and popular {\em mixed} matchings~\cite{KMN09} 
were also proposed to deal with instances that had no popular matchings.

G\"ardenfors~\cite{Gar75}, who introduced the notion of popular matchings, considered
this problem in the domain of 2-sided preference lists. In any instance $G = (A \cup B,E)$ 
with 2-sided strict preference lists, a stable matching is actually a minimum size popular matching
and efficient algorithms for computing a maximum size popular matching were given in~\cite{HK11,Kav12}. 

\medskip

\noindent{\bf Organization of the paper.} Section~\ref{sec:prelims} has preliminaries.
Section~\ref{sec:algo} contains our algorithm and its proof of correctness. Section~\ref{sec:hardness}
shows our $\mathsf{NP}$-hardness result. We conclude with some open problems.

\section{Preliminaries}
\label{sec:prelims}

For any $a \in A$, let $f(a)$ denote $a$'s most desired, first choice post. Let $F = \{f(a): a \in A\}$ be the set of these top posts. We will refer to posts in $F$ as $f$-posts and to those in 
$B \setminus F$ as non-$f$-posts. For any $a \in A$, let $r_a$ be the rank of $a$'s most 
preferred non-$f$-post in $a$'s preference list; when all of $a$'s neighbors are in $F$, we 
set $r_a = \infty$. The following theorem characterizes popular matchings in the 
1-sided voting model. 
\begin{theorem}[from \cite{AIKM05}]
\label{thm1}
Let $G = (A \cup B,E)$ be an instance of the 1-sided popular matching problem, where each 
$a \in A$ has a strict preference list.
Let $M$ be any matching in~$G$. $M$ is popular if and only if the following two properties are 
satisfied:
\begin{enumerate} [(i)]
\item $M$ matches every $b \in F$ to some applicant $a$ such that $b = f(a)$;
\item $M$ matches each applicant $a$ to either $f(a)$ or its neighbor of rank~$r_a$.
\end{enumerate}
\end{theorem}

Thus the only applicants that may be left unmatched in a popular matching here are those 
$a \in A$ that satisfy $r_a = \infty$.

Let us consider the following example where $A = \{a_1,a_2,a_3\}$ and $B = \{b_1,b_2,b_3\}$: 
both $a_1$ and $a_2$ have the same preference list which is $b_1 > b_2$ ($b_1$ followed by 
$b_2$) while $a_3$'s preference list is $b_1 > b_2 > b_3$ (see the top left figure in 
Fig.~\ref{fig:AppA}). Assume first that only applicants cast votes.
The only posts that any of $a_1,a_2,a_3$ can be matched to in a popular matching here 
are $b_1$ and~$b_2$.  As there are three applicants and only two possible partners 
in a popular matching, there is no popular matching here. However in our 2-sided voting model, 
where posts also care about being matched and all neighbors are in a single tie, 
we have a popular matching $\{(a_1,b_1),(a_2,b_2),(a_3,b_3)\}$. Note that $b_3$ is ranked third 
in $a_3$'s preference list, which is worse than $r_{a_3} = 2$, however such edges are permitted 
in popular matchings in our 2-sided model.

Consider the following example (see the middle figure in Fig.~\ref{fig:AppA}): $A = \{a_0,a_1,a_2,a_3\}$ 
and $B = \{b_0,b_1,b_2,b_3\}$; both $a_1$ and $a_2$ have the same preference list which is $b_1 > b_2$ 
while $a_3$'s preference list is $b_1 > b_0 > b_2$ and $a_0$'s preference list is $b_0 > b_3$. There 
is again no popular matching here in the 1-sided model, however in our 2-sided voting model, we have 
a popular matching $\{(a_0,b_3),(a_1,b_1),(a_2,b_2),(a_3,b_0)\}$. Note that $b_0 \in F$ and here it is 
matched to $a_3$ and $f(a_3) \ne b_0$; also $a_3$ is matched to its second ranked post: this 
is neither its top post nor its  $r_{a_3}$-th ranked post ($r_{a_3} = 3$ here).

Thus popular matchings in our 2-sided voting model are quite different from the 
characterization given in Theorem~\ref{thm1} for popular matchings in the 1-sided model. Our 
algorithm (presented in Section~\ref{sec:algo}) uses the following decomposition.

\paragraph{Dulmage-Mendelsohn decomposition~\cite{dulmage}.} Let $M$ be a maximum matching in 
a bipartite graph $G = (A \cup B, E)$. Using $M$, we can partition $A \cup B$ into three 
disjoint sets: a vertex $v$ is \emph {even} (similarly, \emph {odd}) if there is an even 
(resp., odd) length alternating path with respect to $M$ from an unmatched vertex to~$v$.
Similarly, a vertex $v$ is \emph {unreachable} if there is no alternating path from an 
unmatched vertex to~$v$. Denote by $\mathcal E$, $\mathcal O$, and $\mathcal U$ the sets of 
even, odd, and unreachable vertices, respectively. The following properties (proved 
in~\cite{GGL95}) will be used in our algorithm and analysis.
\begin{itemize}
\item $\mathcal E$, $\mathcal O$, and $\mathcal U$ are pairwise disjoint. Let $M'$ be any
maximum matching in $G$ and let $\mathcal E'$, $\mathcal O'$, and $\mathcal U'$ be the sets 
of even, odd, and unreachable vertices with respect to $M'$, respectively. Then ${\mathcal E} 
= {\mathcal E'}$, ${\mathcal O} = {\mathcal O'}$, and ${\mathcal U} = {\mathcal U'}$.
\item Every maximum matching $M$ matches all vertices in ${\mathcal O} \cup {\mathcal U}$ and 
has size $|{\mathcal O}| + |{\mathcal U}|/2$. In $M$, every vertex in $\mathcal O$ is matched 
with some vertex in $\mathcal E$, and every vertex in $\mathcal U$ is matched with another 
vertex in~$\mathcal U$. 
\item The graph $G$ has no edge in ${\mathcal E}\times({\mathcal E}\cup{\mathcal U})$.
\end{itemize}

\section{Finding popular matchings in a 2-sided voting model}
\label{sec:algo}
The input is $G = (A\cup B, E)$ where each applicant $a \in A$ has a strict preference list 
while each post $b \in B$ has a single tie as its preference list. Our algorithm below builds 
a graph $H$ using a partition $\langle X,Y,Z\rangle$ of $B$ that is 
constructed in an iterative manner.
Initialize $X = F$, $Y = B \setminus F$, and $Z = \emptyset$. 

For any $a \in A$, recall that $r_a$ is the rank of $a$'s most preferred non-$f$-post.
For any $U \subseteq B$, let $\Nbr(U)$ (similarly, $\Nbr_H(U)$) denote the set of neighbors in 
$G$ (resp., in $H$) of the vertices in~$U$. Note that our algorithm will maintain
$\Nbr_H(X) \cap \Nbr(Z) = \emptyset$ by ensuring that $\Nbr_H(X) \subseteq A \setminus \Nbr(Z)$.

\begin{enumerate}[(I)]
\item While $\mathsf{true}$ do
\begin{enumerate}[1.]
\setcounter{enumii}{-1}
\item $H$ is the empty graph on $A \cup B$.
\item For each $a \in A \setminus \Nbr(Z)$ do: 

-- if $f(a) \in X$ then add the edge $(a,f(a))$ to~$H$.

\item For every $b \in X$ that is isolated in~$H$ do: 

-- delete $b$ from $X$ and add $b$ to~$Y$.

\item For each $a \in A$ do: 

-- let $b$ be $a$'s most preferred post in the set $Y$; if the rank of $b$ in $a$'s preference 
list is $\le r_a$ (i.e.\ $r_a$ or better), then add $(a,b)$ to~$H$.

\item Consider the graph $H$ constructed in steps~1-3.  Compute a maximum matching in~$H$. 
{\em [This is to identify ``even'' posts in~$H$.]}
\begin{itemize}
\item If there exist even posts in~$Y$ then delete all even posts from $Y$ and add them 
to~$Z$.
\item Else quit the While-loop.
\end{itemize}
\end{enumerate}

\smallskip

\item Every $a \in \Nbr(Z)$ adds the edge $(a,b)$ to $H$ where $b$ is $a$'s most preferred post 
in the set~$Z$. 

\smallskip

\item Add all posts in $D = \{\ell(a): a \in A\ \mathrm{and}\ r_a=\infty\}$ to $Y$, where 
$\ell(a)$ is the {\em dummy} last resort post of applicant~$a$. For every applicant $a$ such that 
$\Nbr(\{a\}) \subseteq X$, add the edge $(a,\ell(a))$ to~$H$. 
\end{enumerate}

Note that
introducing dummy posts does not interfere with the voting for popular matchings because 
dummy posts do not vote -- they are only present in the 
``helper'' graph $H$ constructed above and not in the popular matching instance~$G$. For any
applicant $a$, being matched to $\ell(a)$ is equivalent to $a$ being left unmatched. Thus any 
matching $M$ in $H$ can be projected to a matching in $G$, by deleting all $(a,\ell(a))$ 
edges from $M$ and for convenience, we will refer to the resulting matching also as $M$.

The condition for exiting the While-loop ensures that all posts in $Y$, and hence 
all in $X \cup Y$, are odd/unreachable in the subgraph of $H$ with the set of posts 
restricted to {\em real} posts in $X \cup Y$ (i.e.\ the non-dummy ones). So starting 
with a maximum matching in this subgraph and augmenting it after adding the edges on 
posts in $Z$ in phase~(II) and the edges on dummy posts in phase~(III), we get a maximum matching 
in $H$ that matches all real posts in~$X \cup Y$. After the construction of $H$, our 
algorithm for the popular matching problem in $G$ is given below.
\begin{itemize}
\item If $H$ admits an $A$-complete matching, then return one that matches all real posts 
in~$X \cup Y$; else output ``$G$ has no popular matching''.
\end{itemize}

In the rest of this section, we prove the following theorem.

\begin{theorem} 
\label{thm:correctness}
$G$ admits a popular matching if and only if $H$ admits an $A$-complete matching, i.e.\ one 
that matches all vertices in~$A$.
\end{theorem}

\subsubsection{Some examples.}

We present some examples here and describe how our algorithm builds the graph $H$ on these examples. 
Let $X_i,Y_i,Z_i$ denote the sets 
$X,Y,Z$ at the end of the $i$-th iteration of our algorithm and let $H_i$ denote
the graph $H$ in step~4 of the $i$-th iteration of our algorithm.

\tikzstyle{vertex} = [circle, draw=black, fill=black, inner sep=0pt,  minimum size=5pt]
\tikzstyle{edgelabel} = [circle, fill=white, inner sep=0pt,  minimum size=15pt]
\begin{center}
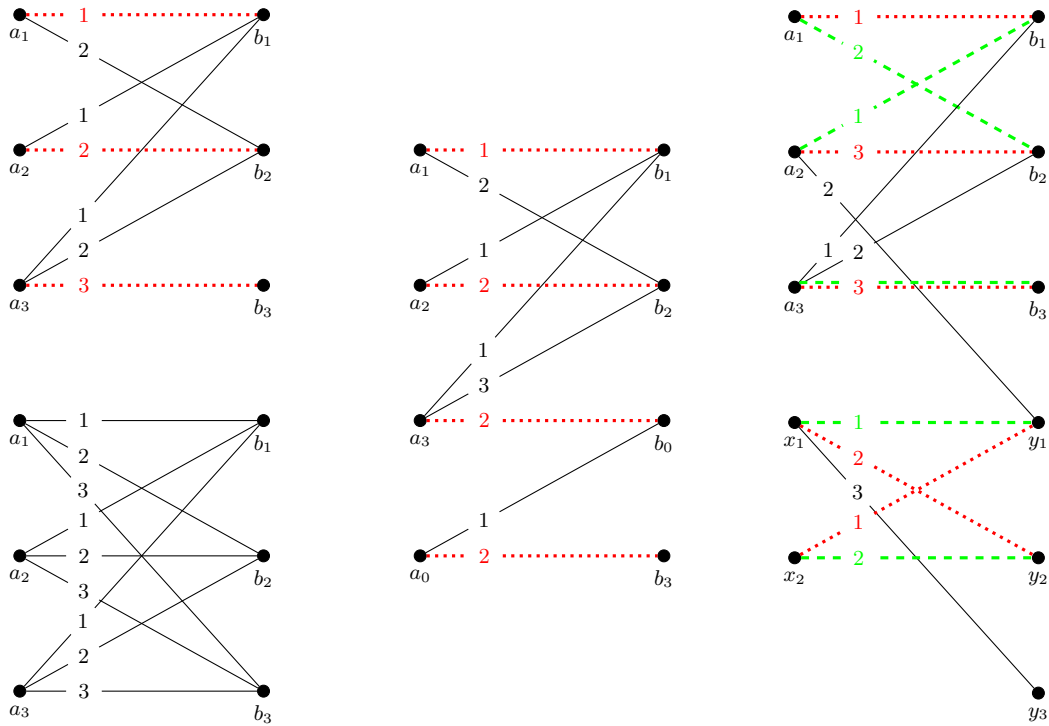
\begin{figure}[h!]
\centering
	\pgfmathsetmacro{\d}{2}
	\pgfmathsetmacro{\b}{3.6}
\begin{minipage}{0.28\textwidth}
\centering
\begin{tikzpicture}[scale=0.9, transform shape]

	\node[vertex, label=below:$a_1$] (a1) at (0,0) {};
	\node[vertex, label=below:$a_2$] (a2) at ($(a1) + (0, -\d)$) {};
	\node[vertex, label=below:$a_3$] (a3) at ($(a1) + (0, -2*\d)$) {};
	\node[vertex, label=below:$b_1$] (b1) at (\b,0) {};
	\node[vertex, label=below:$b_2$] (b2) at ($(b1) + (0, -\d)$) {};
	\node[vertex, label=below:$b_3$] (b3) at ($(b1) + (0, -2*\d)$) {};

	\draw [very thick, red, dotted] (a1) -- node[edgelabel, near start] {1} (b1);
	\draw [] (a1) -- node[edgelabel, near start] {2} (b2);
	\draw [] (a2) -- node[edgelabel, near start] {1} (b1);
	\draw [very thick, red, dotted] (a2) -- node[edgelabel, near start] {2} (b2);
	\draw [] (a3) -- node[edgelabel, near start] {1} (b1);
	\draw [] (a3) -- node[edgelabel, near start] {2} (b2);
	\draw [very thick, red, dotted] (a3) -- node[edgelabel, near start] {3} (b3);

	\node[vertex, label=below:$a_1$] (A1) at ($(a1) + (0, -6)$) {};
	\node[vertex, label=below:$a_2$] (A2) at ($(A1) + (0, -\d)$) {};
	\node[vertex, label=below:$a_3$] (A3) at ($(A1) + (0, -2*\d)$) {};
	\node[vertex, label=below:$b_1$] (B1) at (\b,-6) {};
	\node[vertex, label=below:$b_2$] (B2) at ($(B1) + (0, -\d)$) {};
	\node[vertex, label=below:$b_3$] (B3) at ($(B1) + (0, -2*\d)$) {};

	\draw [] (A1) -- node[edgelabel, near start] {1} (B1);
	\draw [] (A1) -- node[edgelabel, near start] {2} (B2);
	\draw [] (A1) -- node[edgelabel, near start] {3} (B3);
	\draw [] (A2) -- node[edgelabel, near start] {1} (B1);
	\draw [] (A2) -- node[edgelabel, near start] {2} (B2);
	\draw [] (A2) -- node[edgelabel, near start] {3} (B3);
	\draw [] (A3) -- node[edgelabel, near start] {1} (B1);
	\draw [] (A3) -- node[edgelabel, near start] {2} (B2);
	\draw [] (A3) -- node[edgelabel, near start] {3} (B3);

\end{tikzpicture}
\end{minipage}\hspace{0.7cm}\begin{minipage}{0.28\textwidth}
\centering
\begin{tikzpicture}[scale=0.9, transform shape]

	\node[vertex, label=below:$a_1$] (a1) at (0,0) {};
	\node[vertex, label=below:$a_2$] (a2) at ($(a1) + (0, -\d)$) {};
	\node[vertex, label=below:$a_3$] (a3) at ($(a1) + (0, -2*\d)$) {};
	\node[vertex, label=below:$a_0$] (a0) at ($(a1) + (0, -3*\d)$) {};
	\node[vertex, label=below:$b_1$] (b1) at (\b,0) {};
	\node[vertex, label=below:$b_2$] (b2) at ($(b1) + (0, -\d)$) {};
	\node[vertex, label=below:$b_0$] (b0) at ($(b1) + (0, -2*\d)$) {};
	\node[vertex, label=below:$b_3$] (b3) at ($(b1) + (0, -3*\d)$) {};

	\draw [very thick, red, dotted] (a1) -- node[edgelabel, near start] {1} (b1);
	\draw [] (a1) -- node[edgelabel, near start] {2} (b2);
	\draw [] (a2) -- node[edgelabel, near start] {1} (b1);
	\draw [very thick, red, dotted] (a2) -- node[edgelabel, near start] {2} (b2);
	\draw [] (a3) -- node[edgelabel, near start] {1} (b1);
	\draw [] (a3) -- node[edgelabel, near start] {3} (b2);
	\draw [very thick, red, dotted] (a3) -- node[edgelabel, near start] {2} (b0);
	\draw [] (a0) -- node[edgelabel, near start] {1} (b0);
	\draw [very thick, red, dotted] (a0) -- node[edgelabel, near start] {2} (b3);
	
\end{tikzpicture}
\end{minipage}
\hspace{0.7cm}\begin{minipage}{0.28\textwidth}
\begin{tikzpicture}[scale=0.9, transform shape]

	\node[vertex, label=below:$a_1$] (a1) at (0,0) {};
	\node[vertex, label=below:$a_2$] (a2) at ($(a1) + (0, -\d)$) {};
	\node[vertex, label=below:$a_3$] (a3) at ($(a1) + (0, -2*\d)$) {};
	\node[vertex, label=below:$x_1$] (x1) at ($(a1) + (0, -3*\d)$) {};
	\node[vertex, label=below:$x_2$] (x2) at ($(a1) + (0, -4*\d)$) {};
	\node[vertex, label=below:$b_1$] (b1) at (\b,0) {};
	\node[vertex, label=below:$b_2$] (b2) at ($(b1) + (0, -\d)$) {};
	\node[vertex, label=below:$b_3$] (b3) at ($(b1) + (0, -2*\d)$) {};
	\node[vertex, label=below:$y_1$] (y1) at ($(b1) + (0, -3*\d)$) {};
	\node[vertex, label=below:$y_2$] (y2) at ($(b1) + (0, -4*\d)$) {};
	\node[vertex, label=below:$y_3$] (y3) at ($(b1) + (0, -5*\d)$) {};
	
	\draw [very thick, red, dotted] (a1) -- node[edgelabel, near start] {1} (b1);
	\draw [very thick, green, dashed] (a1) -- node[edgelabel, near start] {2} (b2);
	\draw [very thick, green, dashed] (a2) -- node[edgelabel, near start] {1} (b1);
	\draw [very thick, red, dotted] (a2) -- node[edgelabel, near start] {3} (b2);
	\draw [] (a2) -- node[edgelabel, very near start] {2} (y1);
	\draw [] (a3) -- node[edgelabel, very near start] {1} (b1);
	\draw [] (a3) -- node[edgelabel, near start] {2} (b2);
	\draw [very thick, green, dashed] ($(a3) + (0.07, 0.07)$) -- ($(b3) + (-0.08, 0.07)$);
	\draw [very thick, red, dotted] (a3) -- node[edgelabel, near start] {3} (b3);
	\draw [very thick, green, dashed] (x1) -- node[edgelabel, near start] {1} (y1);
	\draw [very thick, red, dotted] (x1) -- node[edgelabel, near start] {2} (y2);
	\draw [] (x1) -- node[edgelabel, near start] {3} (y3);
	\draw [very thick, red, dotted] (x2) -- node[edgelabel, near start] {1} (y1);
	\draw [very thick, green, dashed] (x2) -- node[edgelabel, near start] {2} (y2);
\end{tikzpicture}
\end{minipage}

\caption{We have 4 examples here: except for the graph in bottom left, all the other graphs admit popular matchings and these are highlighted. In the graph on the extreme right, both the red  dotted and green dashed matchings are popular, however the matching $\{(a_1,b_1),(a_2,b_2),(a_3,b_3),(x_1,y_1),(x_2,y_2)\}$ in their union is {\em not} popular.}
\label{fig:AppA}
\end{figure}
\end{center}

In the first example (top left of Fig.~\ref{fig:AppA}), we have $A = \{a_1,a_2,a_3\}$ and $B = \{b_1,b_2,b_3\}$ 
and the preferences of applicants are denoted on the edges.  By our initialization, we have $X_0 = \{b_1\}$, 
$Y_0 = \{b_2,b_3\}$, and $Z_0 = \emptyset$. In step~4 of our first iteration, we identify $b_3$ as an even post in~$H_1$. 
So $Y_1 = \{b_2\}$ and $Z_1 = \{b_3\}$. In the second iteration, $a_3 \in \Nbr(Z_1)$ and so it has no edge to $b_1$ in~$H_2$. 
This is the last iteration of our algorithm. Our final graph $H$ has the edge set 
$\{(a_1,b_1),(a_2,b_1),(a_1,b_2),(a_2,b_2),(a_3,b_2),(a_3,b_3)\}$.

While the above example admits a popular matching,
consider the graph in the bottom left of Fig.~\ref{fig:AppA}. The first iteration of our algorithm is exactly the same
on this graph as it was with the earlier graph. We have $X_1 = \{b_1\}$, $Y_1 = \{b_2\}$, and $Z_1 = \{b_3\}$.
However in the second iteration, all the applicants $a_1,a_2,a_3$ become elements of $\Nbr(Z_1)$ and $b_1$ becomes an isolated vertex
in step~2, so $b_1$ becomes an element of~$Y_2$. In step~4 of the second iteration, $b_2$ is identified as an even post in $H_2$ 
as it is isolated in~$H_2$.  So $Y_2 = \{b_1\}$ and $Z_2 = \{b_2,b_3\}$. No demotions happen in the third iteration, 
which is the last iteration of our algorithm. Our final graph $H$ has the edge set $\{(a_1,b_1),(a_2,b_1),(a_3,b_1),(a_1,b_2),(a_2,b_2),(a_3,b_2)\}$. Observe that $H$ has no 
$A$-complete matching.

\smallskip

In the third example (middle of Fig.~\ref{fig:AppA}), we have $A = \{a_0,a_1,a_2,a_3\}$ and $B = \{b_0,b_1,b_2,b_3\}$ 
and the preferences of applicants are again denoted on the edges. In step~4 of the first iteration of this algorithm, the post $b_3$
is identified as an even vertex in $Y_0$ and it becomes an element of~$Z_1$. So $a_0 \in \Nbr(Z_1)$ and $b_0$ becomes isolated in step~2 
of the second iteration. So $b_0$ becomes an element of $Y_2$  and this is the last iteration of our algorithm. 
Our final graph $H$ has the edge set $\{(a_1,b_1),(a_2,b_1),(a_3,b_1),(a_1,b_2),(a_2,b_2),(a_3,b_0),(a_0,b_0),(a_0,b_3)\}$. 
This graph admits an $A$-complete matching $\{(a_1,b_1),(a_2,b_2),(a_3,b_0),(a_0,b_3)\}$.

\smallskip

The fourth example here (the rightmost graph in Fig.~\ref{fig:AppA}) is that of a  graph $G$ with several popular matchings. 
It is not the case that $H$ contains all these matchings.  At the end of our entire algorithm, we have $X = \{b_1,y_1\}$, $Y = \{b_2\}$,
and $Z = \{b_3,y_2,y_3\}$. The graph $H$ does not contain the edges $(a_3,b_1)$ and $(x_1,y_1)$ since $a_3$ and $x_1$ belong to~$\Nbr(Z)$. 
The subgraph $H$ admits an $A$-complete matching $M = \{(a_1,b_1),(a_2,b_2),(a_3,b_3),(x_2,y_1),(x_1,y_2)\}$ and this is a popular matching in~$G$.
However $H$ does not contain $M' = \{(a_1,b_2),(a_2,b_1),(a_3,b_3),(x_1,y_1)$, $(x_2,y_2)\}$, which is another popular matching in~$G$.
In fact, any subgraph that contains both $M$ and $M'$ would also contain the following $A$-complete matching
$\{(a_1,b_1),(a_2,b_2),(a_3,b_3),(x_1,y_1),(x_2,y_2)\}$, which is {\em not} popular.

\subsection{Proof of Theorem~\ref{thm:correctness}: the sufficient part} 
\label{sec:suff}

We first show that if $H$ admits an $A$-complete matching, then 
$G$ admits a popular matching. We have already observed that if $H$ admits an $A$-complete matching, then $H$ has an $A$-complete matching $M$ that matches all real posts in~$X \cup Y$.

A useful observation is that $Z \subseteq B \setminus F$. This is because in step~4 of the 
While-loop in our algorithm, all $f$-posts in $Y$ are odd/unreachable in $H$ as they are the 
only neighbors in $H$ of applicants who regard them as $f$-posts.

We now assign edge labels in $\{\pm 1\}$ to all edges in~$G \setminus M$:
for an edge $(a,b)$ in $G \setminus M$, if $a$ prefers $b$ to $M(a)$, 
then we label this edge $+1$, else we label this~$-1$. The label of $(a,b)$ is basically $a$'s 
vote for $b$ vs~$M(a)$. Fig.~\ref{fig:next} is helpful here.

\begin{figure}[h]
\centerline{\resizebox{0.65\textwidth}{!}{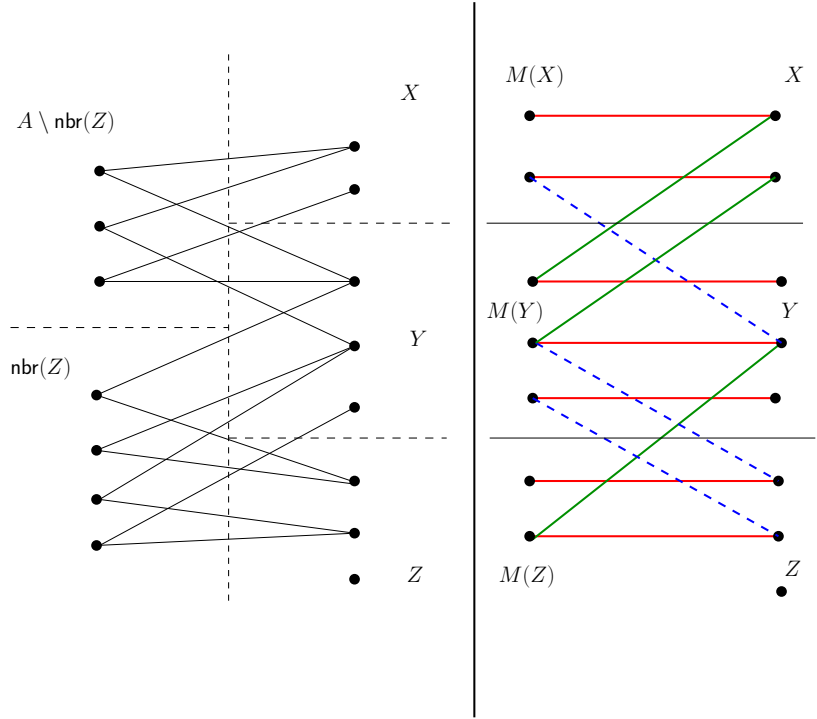}}
\caption{The set $B$ gets partitioned into $X,Y$, and~$Z$.  We have $\Nbr_H(X) \cap \Nbr(Z) = \emptyset$. In the figure on the right, the horizontal edges  belong to~$M$. Only the edges of $(M(Y) \times X) \cup (M(Z)\times(X\cup Y))$ can be labeled~$+1$.}
\label{fig:next}
\end{figure}

For any $U \in \{X, Y, Z\}$, let $M(U) \subseteq A$ be the set of applicants matched in $M$ 
to posts in~$U$. The following lemma is important. 

\begin{lemma}
\label{lem:edge-signs}
Every edge of $G$ in $M(X) \times Y$ is labeled~$-1$; similarly, every edge in $M(Y) \times Z$ 
is labeled~$-1$. Any edge labeled~$+1$ has to be either in $M(Y) \times X$ or in 
$M(Z) \times (X \cup Y)$.
\end{lemma}
\begin{proof}
Every edge of $\Nbr(X)\times X$ that is present in $H$ is a top ranked edge. Since 
$M$ belongs to $H$, the edges of $M$ from $\Nbr(X)\times X$ are top ranked edges. Thus it is 
clear that every edge of $G$ in $M(X) \times Y$ is labeled~$-1$. Regarding $M(Y) \times Z$, 
every edge of $\Nbr(Y) \times Y$ that is present in the graph $H$ is an edge $(a,b)$ where the 
rank of $b$ in $a$'s preference list is $\le r_a$ (i.e.\ $r_a$ or better); on the other hand, 
every edge of $\Nbr(Z) \times Z$ that is present in the graph $H$ is an edge $(a,b')$ 
where the rank of $b'$ in $a$'s preference list is $\ge r_a$ (because $b' \in B\setminus F$). 
Since $M$ belongs to $H$, the edges of $M$ from $\Nbr(Y)\times Y$ are ranked better than the edges 
of~$\Nbr(Z) \times Z$. Thus every edge of $G$ in $M(Y) \times Z$ is labeled~$-1$.

We now show that any edge labeled~$+1$ has to be in either  $M(Y) \times X$ or
$M(Z) \times (X \cup Y)$ (see Fig.~\ref{fig:next}). Consider any edge $(a,b) \notin M$ 
such that $b \in U$ and $a \in M(U)$, where $U \in \{X, Y, Z\}$. It follows from the 
construction of the graph $H$ that a vertex in $\Nbr(U)$ can be adjacent in $H$ to only its most 
preferred post in~$U$. Thus any edge $(a,b) \notin M$ where $b \in U$ and $a \in M(U)$ is 
ranked~$-1$. We have already seen that all edges in $M(X) \times Y$ and in $M(Y) \times Z$ are 
labeled~$-1$. There are no edges in $M(X) \times Z$ since $M(X) \subseteq A \setminus \Nbr(Z)$. 
Thus any edge labeled~$+1$ has to be in either  $M(Y) \times X$ or $M(Z) \times (X \cup Y)$. \qed
\end{proof}

Let $M'$ be any matching in~$G$. The symmetric difference of $M'$ and $M$ is denoted by 
$M' \oplus M$: this consists of alternating paths and alternating cycles -- note that
edges here alternate between $M$ and~$M'$. 
Recall that last resort posts are not used in~$M'$ (which is a matching in $G$) 
whereas last resort posts may be present in~$M$ (which is a matching in $H$). 

\begin{lemma}
\label{lem:1-3}
Consider $M' \oplus M$. The following three properties hold:
\begin{enumerate}[(i)]
\item\label{le2i} in any alternating cycle in  $M' \oplus M$, the number of edges that are 
labeled~$-1$ is at least the number of edges that are labeled~$+1$.
\item\label{le2ii} in any alternating path in $M' \oplus M$, the number of edges that are 
labeled~$+1$ is at most {\em two plus} the number of edges that are labeled~$-1$; in case one 
of the endpoints of this path is a last resort post, then the number of edges labeled~$+1$ is 
at most {\em one plus} the number of edges labeled~$-1$.
\item\label{le2iii} in any even length alternating path in $M' \oplus M$, the number of edges 
that are labeled~$-1$ is at least the number of edges that are labeled~$+1$;
in case one of the endpoints of this path is a last resort post, then
the number of edges labeled~$-1$ is at least  {\em one plus} the number of edges labeled~$+1$.
\end{enumerate}
\end{lemma}
\begin{proof}
\noindent{\em Property~(\ref{le2i}).} Let $C \in M \oplus M'$ be an alternating cycle. Let $C$ be 
$b_0$-$a_0$-$b_1$-$a_1$-$b_2$-$\cdots$-$a_{k-1}$-$b_0$, 
where $(a_i,b_i) \in M$ for $0 \le i \le k-1$. If $C$ contains no vertex of $Z$, then there cannot
be two consecutive non-matching edges labeled~$+1$ in~$C$. That is, if $(a_i,b_{i+1})$ is 
labeled~$+1$, then $b_{i+1} \in X$ 
and there is no $+1$ edge incident on $M(b_{i+1}) = a_{i+1}$, thus the non-matching edge incident on $a_{i+1}$ 
in $C$ has to be labeled~$-1$. Hence the number of edges that are labeled~$-1$ is at least the number of edges 
that are labeled~$+1$.

Suppose $C$ contains a vertex of~$Z$: let $b_i$ be such a vertex. There can be two consecutive non-matching edges labeled~$+1$ now: let $b_i$-$a_i$-$b_{i+1}$-$a_{i+1}$-$b_{i+2}$ be such an alternating path within $C$, where both $(a_i,b_{i+1})$ and $(a_{i+1},b_{i+2})$ are labeled~$+1$. Then $b_i \in Z$, $b_{i+1} \in Y$, and $b_{i+2} \in X$. In the first place, there is no $+1$ edge incident on $a_{i+2}$ and the crucial part is that there is {\em no} edge in $G$ between a vertex in $\Nbr_H(X)$ and a vertex in~$Z$. Thus once we reach a vertex $a_{i+2} \in M(X)$, we have to see an edge labeled  $-1$ and so as to reach a vertex in $Z$, we need to see at least {\em two} consecutive non-matching edges labeled~$-1$. Thus it again follows that the number of edges that are labeled~$-1$ is at least the number of edges that are labeled~$+1$.

\medskip

\noindent{\em Property~(\ref{le2ii}).} Let $\rho \in M \oplus M'$ be an alternating path. Let $\rho$ be
$b_0$-$a_0$-$b_1$-$a_1$-$b_2$-$\cdots$-$a_{k-1}$-$b_k$-$a_k$, where $(a_i,b_i) \in M$ for 
$0 \le i \le k$. The same argument that was used in the proof of property~(\ref{le2i}) shows us that 
there can be at most two consecutive non-matching edges labeled~$+1$ in $\rho$ and once we 
traverse such an alternating path $b_i$-$a_i$-$b_{i+1}$-$a_{i+1}$-$b_{i+2}$ in $\rho$ (where $b_i$ 
has to be in $Z$), we are at a vertex~$b_{i+2} \in X$. Thereafter we have to see at least two more 
non-matching edges 
labeled~$-1$ than those labeled~$+1$ to again reach a vertex in~$Z$. Thus it follows that the 
difference between the number of edges that are labeled~$+1$ and the number of edges that are 
labeled~$-1$ is at most two.
 
In fact, for the difference between the number of edges that are labeled~$+1$ and the number of 
edges that are labeled~$-1$ to be exactly two, it has to be the case that $b_0$ is in~$Z$. For, in 
case $b_0$ is in $Y$, then it is easy to see that the difference between the number of 
non-matching edges that are labeled~$+1$ and the number of non-matching edges that are labeled 
$-1$ is at most one. Note that all last resort posts belong to~$Y$. Thus when $b_0$ is a last 
resort post, then the number of edges labeled~$+1$ in $\rho$ is at most {\em one plus} the 
number of edges labeled~$-1$.

\medskip

\noindent{\em Property~(\ref{le2iii}).} Let $\rho = b_0$-$a_0$-$b_1$-$a_1$-$b_2$-$\cdots$-$a_{k-1}$-$b_k$ be an 
even length alternating path where $(a_i,b_i) \in M$ for $0 \le i \le k-1$. The post $b_0$ is unmatched in $M'$ 
and $b_k$ is unmatched in~$M$. Recall that $M$ is $A$-complete, thus any even length alternating path with 
respect to $M$ has to have vertices in $B$ as its endpoints (since one of them is left unmatched in~$M$). 
Since $b_k$ is a post that is matched in $M'$ but not in $M$, it follows that $b_k \in Z$ (as all non-dummy 
posts in $X \cup Y$ are matched in~$M$).

Now the argument is similar to the proof of property~(\ref{le2ii}). In order to maximize the difference between 
the number of edges labeled~$+1$ and those labeled~$-1$, we assumed that the starting vertex $b_0 \in Z$. For the 
final vertex $b_k$ to be in $Z$, it follows that the number of edges that are labeled~$-1$ is at least the number 
of edges that are labeled~$+1$. In particular, when $b_0$ is a last resort post, then the starting vertex is in $Y$ 
and so the number of edges that are labeled~$-1$ is at least {\em one plus} the number of edges that are labeled~$+1$. \qed
\end{proof}

Lemma~\ref{new-cor} uses the above lemma to show the popularity of~$M$. This completes the proof that if $H$ admits an 
$A$-complete matching then $G$ admits a popular matching.

\begin{lemma}
\label{new-cor}
For any matching $M'$ in $G$, we have $\phi(M,M') \ge \phi(M',M)$.
\end{lemma}
\begin{proof}
Recall that $M$ is $A$-complete (where some of the posts used in $M$ 
can be last resort posts). For any $a \in A$ and any neighbor $b$ of $a$ such that $b \ne M(a)$, let $\vote_a(b,M(a))$ 
be the label of the edge $(a,b)$, which is $a$'s vote for $b$ vs~$M(a)$. Consider $M \oplus M'$. We will now investigate each component of $M \oplus M'$ -- being a cycle, an odd path or an even path -- and show $\phi(M,M') \ge \phi(M',M)$ for each of them.
\begin{itemize}
\item For any alternating cycle $C \in M \oplus M'$, among the vertices of $C$, the difference between those who prefer 
$M'$ and those who prefer $M$ is equal to $\sum_{(a,M'(a)) \in C}\vote_a(M'(a),M(a))$. It follows from property~(\ref{le2i}) 
that this sum is at most 0.

\item Consider any odd length alternating path $\rho \in M \oplus M'$: its endpoints are an applicant $a'$ and a post $b'$ that are 
unmatched in~$M'$. Assume $b'$ is a non-dummy post. Then among the vertices of $\rho$ that are matched in $M'$, the difference between 
those who prefer $M'$ and those who prefer $M$ is equal to $\sum_{(a,M'(a)) \in \rho}\vote_a(M'(a),M(a))$. It follows from 
property~(\ref{le2ii}) that this sum is at most~2. The two vertices $a'$ and $b'$ prefer $M$ to $M'$ as they are matched in $M$ and 
unmatched in $M'$, since $a'$ is unmatched in~$M'$. Thus summed over all vertices of $\rho$, the difference between those who prefer $M'$ and those who prefer $M$ is 
again at most 0.

Now suppose $b'$ is a dummy post. Then it follows from property~(\ref{le2ii}) that among the vertices of $\rho$ that are matched in 
$M'$, the difference between those who prefer $M'$ and those who prefer $M$ is at most 1. The vertex $a'$ prefers $M$ to~$M'$. Thus 
summed over all real vertices of $\rho$, the difference between those who prefer $M'$ and those who prefer $M$ is again at most 0.

\item Consider any even length alternating path $\rho \in M \oplus M'$: its endpoints are a post $b_0$ that is unmatched in $M'$ 
and a post $b_k$ that is unmatched in~$M$. Assume $b_0$ is a non-dummy post. Then summed over all vertices of $\rho$ 
(this includes $b_0$ who prefers $M$ and $b_k$ who prefers $M'$), the difference between those who prefer $M$ and those who prefer $M'$ 
is at least 0 (by property~(\ref{le2iii})). 

Now suppose $b_0$ is a dummy post.  Then summed over all real vertices of $\rho$ that are matched in $M$, the difference between those 
who prefer $M$ and those who prefer $M'$ is at least 1 (by property~(\ref{le2iii})). Thus summed over all real vertices of $\rho$ 
(this includes $b_k$ who prefers $M'$), the difference between those who prefer $M$ and those who prefer $M'$ is at least 0.
\end{itemize}
All vertices whose partners in $M$ and in $M'$ are different belong to some alternating path or cycle in~$M \oplus M'$. Hence 
the difference between the number of vertices that prefer $M$ and those that prefer $M'$ is non-negative. In other words, 
$\phi(M,M') \ge \phi(M',M)$. \qed
\end{proof}

\subsection{Proof of Theorem~\ref{thm:correctness}: the necessary part}
We now show the other side of Theorem~\ref{thm:correctness}. That is, if $G$ admits a popular 
matching, then $H$ admits an $A$-complete matching. Let $M^*$ be a popular matching in~$G$. 
Lemma~\ref{lem:pop-mat} will be useful to us.
\begin{lemma}
\label{lem:pop-mat}
If $(a,b) \in M^*$ and $b \in F$, then $b$ has rank better than $r_a$ in $a$'s preference list.
\end{lemma}
\begin{proof}
Suppose $(a,b) \in M^*$, where $b \in F$, and $b$ has rank worse than $r_a$ in $a$'s preference list. Note that the rank of $b$ cannot be exactly $r_a$ since there is another 
post $b' \notin F$ that has rank $r_a$ in $a$'s preference list. 
We know that $a = M^*(b)$ prefers post 
$b'$ to~$b$. If post $b'$ is unmatched, then consider $M^* \oplus p$ where $p = M^*(a_0)$-$a_0$-$b$-$a$-$b'$, where $a_0$ is an applicant such that $f(a_0) = b$ (there exists such an applicant since $b \in F$). The matching $M^* \oplus p$ is more popular than~$M^*$. 

So suppose the post $b'$ is matched and let $a_1 = M^*(b')$. If $a_0 = a_1$, then consider the alternating cycle 
$C = a_0$-$b$-$a$-$b'$-$a_0$; the matching $M^* \oplus C$ makes $a_0$ and $a$ swap their partners and both applicants prefer
$M^* \oplus C$ to $M^*$ while nobody prefers $M^*$ to~$M^* \oplus C$. Thus $M^* \oplus C$ is more popular than~$M^*$.
If $a_0 \ne a_1$, then consider the alternating path $\rho$ that promotes $a_0$ to its top post $b$ and then $M^*(b) = a$ to a more 
preferred post $b'$ and finally, $M^*(b') = a_1$ to its top post~$f(a_1)$.
The vertices $M^*(a_0)$ and $M^*(f(a_1))$ become unmatched in $M^*\oplus\rho$ and so they prefer $M^*$ 
to  $M^*\oplus\rho$ while the three vertices $a_0,a$, and $a_1$  prefer $M^*\oplus\rho$ to~$M^*$. 
Every other vertex is indifferent between $M^*\oplus\rho$ and~$M^*$.
So $M^*\oplus\rho$ is more popular than~$M^*$. 
Thus we have contradicted the popularity of $M^*$ in all the cases. \qed
\end{proof}

Analogous to Section~\ref{sec:suff}, we label the edges of $G \setminus M^*$ by $+1$ or~$-1$: the label of an edge $(a,b)$ in 
$G \setminus M^*$ is the vote of $a$ for $b$ vs~$M^*(a)$. In case $a$ is not matched in $M^*$, 
then $\vote(a,b) = +1$ for any neighbor $b$ of~$a$. Due to the popularity of $M^*$, 
the following two properties hold on these edge labels (otherwise $M^* \oplus \rho \succ M^*$).

\begin{enumerate} 
\item[$\mathsf{(i)}$] There is no alternating path $\rho$ such that the edge labels in 
$\rho \setminus M^*$ are $\langle +1, +1, +1, \cdots \rangle$, i.e.\
no three consecutive non-matching edges are labeled~$+1$. 
\item[$\mathsf{(ii)}$]  There is no alternating path $\rho$ where the edge labels in 
$\rho \setminus M^*$ are $\langle +1, +1, -1$, $+1, +1, \cdots \rangle$, i.e.\
no five consecutive non-matching edge labels add up to~4.
\end{enumerate}

Based on the matching $M^*$ and the edge labels on $G \setminus M^*$, we partition $B$ into 
$L_1 \cup L_2 \cup L_3$. 
\begin{itemize}
\item Roughly speaking, $L_3$ consists of unwanted posts, so
all posts that are unmatched in $M^*$ belong to $L_3$. Similarly, posts like $b_3$ with a
length-5 alternating path $M^*(b_1)$-$b_1$-$M^*(b_2)$-$b_2$-$M^*(b_3)$-$b_3$ incident on them,
with both the non-matching edges labeled~$+1$ (see Fig.~\ref{fig:partition}) are in $L_3$.
\item Top posts are split across $L_1$ and $L_2$: property~$\mathsf{(ii)}$ indicates
that applicants matched to posts in $L_1$ should not be adjacent to posts in $L_3$,
hence those top posts whose partners have neighbors in $L_3$ are in~$L_2$ and
the top posts whose partners have {\em no} neighbors in $L_3$ are in~$L_1$.
\end{itemize}

\begin{figure}[h]
\centerline{\resizebox{0.38\textwidth}{!}{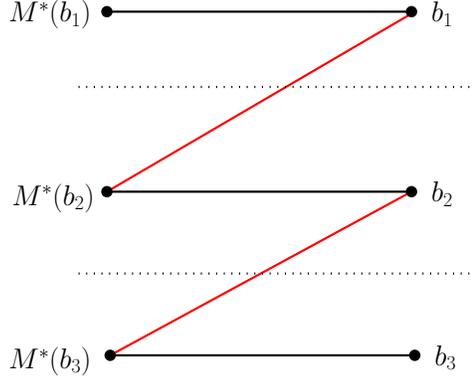}}
\caption{A length-5 alternating path $M^*(b_1)$-$b_1$-$M^*(b_2)$-$b_2$-$M^*(b_3)$-$b_3$, where both $(M^*(b_3),b_2)$ and $(M^*(b_2),b_1)$ are labeled~$+1$.}
\label{fig:partition}
\end{figure}

More formally, we define the partition $B = L_1 \cup L_2 \cup L_3$ below.
\begin{enumerate}\setcounter{enumi}{-1}
\item Initialize $L_1 = L_2 = \emptyset$ and $L_3 = \{b\in B: b$ is unmatched in~$M^*\}$.
We now add more posts to the sets $L_1, L_2, L_3$ as described below.

\item 
For each length-5 alternating path $\rho = a_1$-$b_1$-$a_2$-$b_2$-$a_3$-$b_3$ where 
$(a_1,b_1), (a_2,b_2), (a_3,b_3) \in M^*$ and both $(a_2,b_1)$ and $(a_3,b_2)$ are marked~$+1$, 
add $b_i$ to~$L_i$, for $i = 1,2,3$.

\item Now consider those $b \in B$ that are matched in $M^*$ but $b$ is not a part of any 
length-5 alternating path where both the non-matching edges are labeled~$+1$. We repeat the 
following two steps till there are no more posts to be added to either $L_2$ or~$L_3$
via these rules: 
\begin{itemize}
\item suppose $M^*(b)$ has no $+1$ edge incident on it: if $M^*(b) \in \Nbr(L_3)$, then add 
$b$ to~$L_2$. 
\item if $M^*(b)$ has a $+1$ edge to a vertex in $L_2$, then add $b$ to~$L_3$. 
\end{itemize}
\item For each $b$ such that $M^*(b)$ has no $+1$ edge incident on it: 
\begin{itemize}
\item if $M^*(b) \notin \Nbr(L_3)$, then add $b$ to~$L_1$. 
\end{itemize}
\item For each $b$ not yet in $L_2\cup L_3$ and $M^*(b)$ has a $+1$ edge to a vertex in $L_1$: 
\begin{itemize}
\item add $b$ to~$L_2$. 
\end{itemize}
\end{enumerate}

\begin{lemma}
\label{lemma-F} 
Recall that $F$ is the set of top posts. The above partition $\langle L_1,L_2,L_3\rangle$ satisfies the following properties:
\begin{enumerate}[1.]
\item $F \subseteq L_1 \cup L_2$;
\item $M^*(L_1) \cap \Nbr(L_3) = \emptyset$.
\end{enumerate}
\end{lemma}
\begin{proof}
Suppose $b_1 = f(a_0)$ belongs to~$L_3$. The post $b_1$ has to be matched 
in~$M^*$. Let $a_1 = M^*(b_1)$ and we also know from the construction of the set $L_3$ that 
there is an edge $(a_1,b_2)$ with $b_2 \in L_2$ that is labeled~$+1$. If the vertex $a_0$ is 
unmatched in $M^*$, then by promoting $a_0$ to $b_1$ and $a_1$ to $b_2$, 
and leaving $M^*(b_2)$ unmatched, we obtain a matching that is more popular than~$M^*$.

Hence let us assume that $a_0$ is matched in~$M^*$. Consider the alternating path 
$M^*(a_0)$-$a_0$-$b_1$-$a_1$-$b_2$ with respect to~$M^*$: this has two consecutive non-matching 
edges that are labeled~$+1$. Thus it follows from our construction of $L_1,L_2,L_3$ that 
$M^*(a_0) \in L_3$, $b_1 \in L_2$, and $b_2 \in L_1$. 
This contradicts our assumption that $b_1 \in L_3$. This finishes the proof of (1).

\smallskip

We first state the following claim, which will be used in the proof of (2). We will assume Claim~\ref{lem:alt-path}
and finish the proof of (2) and then prove this claim.
\begin{new-claim}
\label{lem:alt-path}
If $a \in \Nbr(L_3)$ and $M^*(a) = f(a)$, then there is an alternating path $\rho_a$ with respect to $M^*$ 
with $a$ as an endpoint such that either $\rho_a$ is even and  the edge labels on $\rho_a \setminus M^*$ 
are $\langle -1,+1,-1,\cdots,+1,+1\rangle$ or $\rho_a$ is odd and the edge labels on $\rho_a \setminus M^*$ are 
$\langle -1,+1,-1,\cdots,+1,-1\rangle$ where the last edge is incident on an unmatched post.
\end{new-claim}

We now use Claim~\ref{lem:alt-path} to show that $M^*(L_1) \cap \Nbr(L_3) = \emptyset$.
Posts get added to $L_1$ in steps~1 and~3 of the partition scheme. 
Let $b_1$ be a post that got added to $L_1$ in step~1 -- then there is an alternating path 
$p = a_3$-$b_2$-$a_2$-$b_1$-$a_1$ where both $(a_3,b_2)$ and $(a_2,b_1)$ are labeled~$+1$. 
We know from property~$\mathsf{(i)}$ that $b_1 = f(a_1)$; if $a_1 \in \Nbr(L_3)$, then there is 
an alternating path $\rho_{a_1}$ as described in Claim~\ref{lem:alt-path}. If the posts $a_2$ and 
$a_3$ do not appear in $\rho_{a_1}$, then 
consider the alternating path $p'$ which consists of $p$ followed by~$\rho_{a_1}$. It is easy 
to see that $M^* \oplus p'$ is more popular than $M^*$: a contradiction to the popularity of~$M^*$. 

In case $a_2$ appears in $\rho_{a_1}$, then we have an alternating cycle $C$, which is $\rho_{a_1}$ 
truncated till the vertex $a_2$ followed by $a_2$-$b_1$-$a_1$. This cycle has a stretch of 
alternating $-1$ and $+1$ labeled non-matching edges along with two consecutive non-matching 
edges labeled~$+1$: these are the edge $(a_2,b_1)$ and the edge incident on $b_2$ in $\rho_{a_1}$ from a vertex 
in~$M^*(L_3)$. Thus $M^* \oplus C$ is more popular than $M^*$: a contradiction again. If $a_3$ 
appears in  $\rho_{a_1}$, then we can again construct an alternating cycle $C'$ 
(using the $a_1 \leadsto a_3$ subpath of $\rho_{a_1}$ followed by the alternating path $p$).
The matching $M^* \oplus C'$ is more popular than $M^*$ since $C'$ 
has more $+1$ labeled non-matching edges than $-1$ labeled non-matching edges. This again
contradicts the popularity of~$M^*$.

Regarding posts added to $L_1$ in step~3, we add any post $b$ to $L_1$ only after checking that 
$M^*(b) \notin \Nbr(L_3)$. This completes the proof that $M^*(L_1) \cap \Nbr(L_3) = \emptyset$. \qed
\end{proof}

\paragraph{Proof of Claim~\ref{lem:alt-path}.}
Posts are added to $L_3$ in steps 1, 2 and~3. We now study each of these cases. The set $L_3$ was initialized to the set of posts left unmatched in~$M^*$. So at the end of step~0, it is the case that every $a \in \Nbr(L_3)$ has an odd alternating path, which is in fact an edge $(a,b)$ labeled~$-1$, whose one endpoint is $a$ and the
other endpoint is an unmatched post~$b$.

Let $b_3$ be a post that got added to $L_3$ in step~1. Then there is an alternating path 
$b_3$-$a_3$-$b_2$-$a_2$-$b_1$-$a_1$ such that $(a_i,b_i) \in M^*$ for $i = 1,2,3$, and both $(a_3,b_2)$ and 
$(a_2,b_1)$ are marked~$+1$. Thus every neighbor $a \in \Nbr(\{b_3\})$ with $M^*(a) = f(a)$ has an even 
length alternating path $\rho_a = a$-$b_3$-$a_3$-$b_2$-$a_2$-$b_1$-$a_1$ where the edge labels on $\rho_a \setminus M^*$ 
are $\langle -1,+1,+1\rangle$. Note that $a \ne a_1$ -- otherwise $\rho_a$ is an alternating cycle and
$M^* \oplus \rho_a$ is more popular than~$M^*$.

Thus the claim that every $a \in \Nbr(L_3)$ with $M^*(a) = f(a)$ has a desired alternating path $\rho_a$ is true at
the end of step~1. Let $b_3$ be a post that got added to $L_3$ in step~2 and
let us assume that till the point $b_3$ gets added to $L_3$, the claim holds.
Since $b_3$ was added to $L_3$ in step~2, this was due to a $+1$ edge between $a_3 = M^*(b_3)$ 
and a post $b_2 \in L_2$ whose partner $a_2 = M^*(b_2)$ regards $b_2$ as a top post. 
The post $b_2 \in L_2$ because its partner $a_2 \in \Nbr(L_3)$. 
This means there is a desired alternating path $\rho_{a_2}$ incident on~$a_2$.
Neither $b_3$ nor $b_2$ lies on $\rho_{a_2}$ since all the posts in $\rho_{a_2}$ that belong to 
$L_2 \cup L_3$ were added to $L_2 \cup L_3$ prior to $b_2$ joining $L_2$ and $b_3$ joining~$L_3$. 
Consider any neighbor $a$ of $b_3$ that is in $\Nbr(L_3)$ because $b_3 \in L_3$ and $M^*(a) = f(a)$. 
The desired alternating path $\rho_a$ is $a$-$b_3$-$a_3$-$b_2$-$a_2$ followed by~$\rho_{a_2}$. \qed

\medskip

We will use the partition $\langle L_1,L_2,L_3 \rangle$ of $B$ 
to build the following subgraph $G' = (A\cup B, E')$ of~$G$. For each $a \in A$, 
include the following edges in~$E'$:  
\begin{itemize}
\item[(i)] if $a \notin \Nbr(L_3)$, then add the edge $(a,f(a))$ to~$E'$.
\item[(ii)] if $a$ has a neighbor of rank $\le r_a$ in $L_2$, then add the edge $(a,b)$ to $E'$, 
where $b$ is $a$'s most preferred neighbor in~$L_2$.
\item[(iii)] if $a \in \Nbr(L_3)$, then add the edge $(a,b)$ to $E'$, where $b$ is $a$'s most 
preferred neighbor in~$L_3$.
\end{itemize}

\begin{lemma}
\label{lem:main}
Every edge of the matching $M^*$ belongs to the graph~$G'$.
\end{lemma}
\begin{proof}
The set $B$ has been partitioned into $L_1 \cup L_2 \cup L_3$. We will now show that for each post $b_0$ that is matched in $M^*$, the edge $(M^*(b_0),b_0)$ belongs to~$G'$. We distinguish three cases: $b_0 \in L_1$, $b_0 \in L_2$ and $b_0 \in L_3$.

\smallskip

\noindent{\em -- Case~1.} The post $b_0 \in L_1$. Hence there is no $+1$ edge incident on 
$a_0 = M^*(b_0)$, in other words, $b_0 = f(a_0)$.
Lemma~\ref{lemma-F}.2 tells us that $M^*(L_1) \cap \Nbr(L_3) = \emptyset$; 
hence $a_0$ has no neighbor in $L_3$ and by rule~(i) above, 
the edge $(a_0,f(a_0)) = (a_0,b_0)$ belongs to the edge set of~$G'$.

\smallskip

\noindent{\em -- Case~2.} Next we consider the case when $b_0 \in L_2$.
It is easy to see that $b_0$ has to be $a_0$'s most preferred post in $L_2$, where 
$a_0 = M^*(b_0)$. 
Otherwise there would have been an edge $(a_0,b_1)$ labeled~$+1$ with $b_1 \in L_2$, 
where $b_1$ is $a_0$'s most preferred post in~$L_2$. Then either $b_1 \in L_1$ or
$b_0 \in L_3$ (from how we construct the sets $L_1, L_2, L_3$), a contradiction.
We now have to show that the rank of $b_0$ in $a_0$'s preference list is $\le r_a$, 
otherwise the edge $(a_0,b_0)$ does not belong to~$G'$.

Suppose $b_0 \in F$. Since the edge $(a_0,b_0) \in M^*$, which is a popular matching, it follows 
from Lemma~\ref{lem:pop-mat} that $b_0$ is ranked better than $r_{a_0}$ in $a_0$'s preference 
list; thus the edge $(a_0,b_0)$ belongs to~$G'$. So the case left is when $b_0 \notin F$. 
If $b_0$ is not $a_0$'s most preferred post outside $F$, then there is the length-5 alternating 
path $\rho =b_0$-$a_0$-$b_1$-$a_1$-$f(a_1)$-$M^*(f(a_1))$, 
where $b_1$ is the most preferred post of $a_0$ outside $F$ and $a_1 = M^*(b_1)$.
The alternating path $\rho$ has two consecutive non-matching edges $(a_0,b_1)$ and 
$(a_1,f(a_1))$ that are labeled~$+1$. This contradicts the presence of $b_0$ in $L_2$ as such 
a post would have to be in~$L_3$.
Thus if $b_0 \notin F$, then $b_0$ has to be $a_0$'s most preferred post outside $F$, i.e.\
$b_0$ has rank $r_{a_0}$ in $a_0$'s preference list.

\smallskip

\noindent{\em -- Case 3.} We finally consider the case when the post $b_0 \in L_3$. 
We need to show that $b_0$ is the most preferred post of $a_0 = M^*(b_0)$ in~$L_3$. 
Suppose not. 
Let $b_1$ be $a_0$'s most preferred post in~$L_3$. Since $b_1 \in L_3$ while $F \cap L_3 = 
\emptyset$ (by Lemma~\ref{lemma-F}.1), we know that there is an 
edge labeled~$+1$ incident on $a_1 = M^*(b_1)$. Let this edge be~$(a_1,b_2)$ and let $a_2$ 
be~$M^*(b_2)$. So there is a length-5 alternating path $p = b_0$-$a_0$-$b_1$-$a_1$-$b_2$-$a_2$ 
where both the non-matching edges $(a_0,b_1)$ and $(a_1,b_2)$ are labeled~$+1$. This contradicts 
the presence of $b_1$ in $L_3$ as such a post would have to be in~$L_2$.
Thus $b_0$ is $a_0$'s most preferred post in~$L_3$. \qed
\end{proof}

The following lemma shows the relationship between the partition $\langle L_1,L_2,L_3\rangle$ and the partition $\langle X,Y,Z\rangle$ constructed by our algorithm that builds the graph~$H$.
\begin{lemma}
\label{lem:L1andX}
The set $X \supseteq L_1$ and the set $Z \subseteq L_3$.
\end{lemma}
\begin{proof}
In our algorithm that constructs the graph $H$ and the partition $\langle X, Y, Z\rangle$, the 
set $X$ is initialized to $F$ and the set $Y$ is initialized to~$B \setminus F$. As our algorithm 
progresses, in each iteration of the While-loop, some $f$-posts get demoted from $X$ to $Y$ and similarly, 
some non-$f$-posts get demoted from $Y$ to $Z$ till there is an iteration (say, iteration $h+1$) where all 
posts in $Y$ are odd/unreachable in $H$ -- this is the last iteration of the While-loop. 

For any $1 \le k \le h+1$,  let $T_k$ (similarly, $F_k$) be the set of posts that got 
demoted from $Y$ to $Z$ (resp., $X$ to $Y$) in the $k$-th iteration of the While-loop in our algorithm. We have $T_{h+1} = \emptyset$.

Note that $F_1 = \emptyset$ since $Z$ is initialized to $\emptyset$, so
in the first iteration of our algorithm, {\em every} $f$-post $b$ has a neighbor $a \in A\setminus\Nbr(Z)$ such that $f(a) = b$. 
Thus no post is demoted from $X$ to $Y$ in the first iteration. 

The graph $H_1$ is the subgraph of $G$ where each $a \in A$ has at most two neighbors: 
its top post and when $r_a < \infty$, its neighbor of rank~$r_a$. 
Let $S$ be the set of posts in $B \setminus F$ that are odd/unreachable in the graph~$H_1$. 
The set $T_1$ is the set of even non-$f$-posts in~$H_1$,
i.e.\ $T_1 = B \setminus (F \cup S)$. 

We will use the following claims and finish the proof of this lemma
(the proofs of Claims~\ref{lemma-T}-\ref{lem:Z1} are given after the proof of Lemma~\ref{lem:L1andX}).

\begin{new-claim}
\label{lemma-T}
The set $T_1 \subseteq L_3$.
\end{new-claim}

\begin{new-claim}
\label{lem:F1}
For any $1 \le k \le h$, if $\bigcup_{i=1}^k T_i \subseteq L_3$ then $F_{k+1} \subseteq L_2$.
\end{new-claim}

\begin{new-claim}
\label{lem:Z1}
For any $2 \le k \le h$, if $\bigcup_{i=2}^k F_i \subseteq L_2$ then $T_k \subseteq L_3$. 
\end{new-claim}

Claim~\ref{lemma-T} tells us that $T_1 \subseteq L_3$. We now use Claims~\ref{lem:F1} and \ref{lem:Z1} 
alternately to conclude that for
every $1 \le k \le h$, we have $\cup_{i=2}^{k+1} F_i \subseteq L_2$ and $\cup_{i=1}^k T_i \subseteq L_3$.

\smallskip

Thus the set $Z = \cup_{i=1}^h T_i$ is a subset of $L_3$ and the set $F \setminus X = \cup_{i=2}^{h+1} F_i$ 
is a subset of~$L_2$. Since $F \setminus X \subseteq L_2$, it follows that $X \supseteq  F \setminus L_2 \, = \, L_1$. \qed
\end{proof}

\paragraph{Proof of Claim~\ref{lemma-T}.}
Any post in $T_1$ that is left unmatched in $M^*$ has to belong to~$L_3$. Similarly, any $b_0 \in T_1$ that is
matched to an applicant $a_0$ that ranks $b_0$ worse than $r_{a_0}$ has to belong to~$L_3$: this is 
because there is a length-5 alternating path $p = b_0$-$a_0$-$b_1$-$a_1$-$b_2$-$a_2$ where $b_1$ is a post of rank $r_{a_0}$ 
in $a_0$'s preference list, $a_1 = M^*(b_1)$, and $b_2 = f(a_1)$. The path $p$ has
two consecutive non-matching edges that are labeled~$+1$, so $b_0 \in L_3$.

Now consider any $b_0 \in T_1$ that is matched in $M^*$ to an applicant $a_0$ such that the rank of 
$(a_0,b_0)$ is~$r_{a_0}$. So $a_0$ is a neighbor of $b_0$ in~$H_1$.
Since $b_0$ is even in $H_1$, all the neighbors of $b_0$ in $H_1$ are odd and thus they have to be of
degree {\em exactly} 2 in $H_1$ (recall that  all applicants have degree at most 2 in $H_1$).
Thus the neighbors of these applicants are again even. Let  $C$ be the connected component containing $b_0$ 
in~$H_1$. It is easy to see that in $C$, all posts are even, all applicants are odd, and the
number of posts is more than the number of applicants. (In fact, $C$ is a tree with $b_0$ as the root and  
the number of posts in $C$ is {\em one} plus the number of applicants in~$C$.)

If $b_0 \in L_2$, then $a_0$'s other neighbor in $C$, which is $f(a_0)$, has to be in $L_1$ since
there is a $+1$ edge from $a_0$ to~$f(a_0)$. This means $f(a_0)$ is matched to an applicant $a'_0$ that ranks
it as a top post, so the applicant $a'_0$ is a neighbor of $f(a_0)$ in~$C$. There has to be another
neighbor of $a'_0$ in $C$, call this~$b_1$. The important observation is that $b_1$ cannot be in $L_3$ as that would
violate Lemma~\ref{lemma-F}.2 since $a'_0 \in M^*(L_1)$. So $b_1 \in L_2$ and this means $b_1$ is matched to an
applicant $a_1$ that ranks it $r_{a_1}$, in other words, $a_1$ is a neighbor of $b_1$ in~$C$. So $f(a_1)$ has to be in $L_1$ 
and we continue in this manner marking all $f$-posts in $C$ as elements of $L_1$ and all non-$f$-posts in $C$ as elements 
of~$L_2$. 

This means all posts in $C$ are matched to their neighbors in $C$, however this is not possible as there are more posts 
than applicants in~$C$.  This contradicts our assumption that $b_0 \in L_2$, in other words, $b_0$ has to be in~$L_3$. Thus 
$T_1 \subseteq L_3$. \qed

\paragraph{Proof of Claim~\ref{lem:F1}.}
The set $F_{k+1}$ is the set of posts that got 
demoted from $X$ to $Y$ in the $(k+1)$-th iteration of the While-loop: this means each post $b$ 
in $F_{k+1}$ had no applicant outside $\Nbr(\cup_{i=1}^k T_i)$ that regarded $b$ as an $f$-post. 
In other words, every applicant $a$ such that $f(a) = b$ belongs to $\Nbr(\cup_{i=1}^k T_i)$. 
Since $\cup_{i=1}^k T_i \subseteq L_3$, each such applicant $a$ is present in~$\Nbr(L_3)$. 

\smallskip

Let $F_{k+1} = \{b_1,\ldots,b_h\}$. For $1 \le i \le h$, let $(a_i,b_i) \in M^*$: 
if $f(a_i) = b_i$, then $b_i \in L_2$ (because $a_i \in \Nbr(L_3)$); 
else there is an edge $(a_i,f(a_i))$ that is labeled~$+1$ incident on $a_i$ and hence $b_i$ 
cannot be in~$L_1$. Thus $F_{k+1} \cap L_1 = \emptyset$, i.e.\ $F_{k+1} \subseteq L_2$ (by 
Lemma~\ref{lemma-F}.1). \qed

\paragraph{Proof of Claim~\ref{lem:Z1}.}
Let us assume that we have proved Claim~\ref{lem:Z1} for all smaller values of~$k$. That is, for $j \le k-1$, we have
shown that if $\cup_{i=2}^j F_i \subseteq L_2$ then the set $T_j \subseteq L_3$. This is indeed the case for $k = 2$
since we know $T_1 \subseteq L_3$ (by Claim~\ref{lemma-T}).
Using Claim~\ref{lem:F1} and Claim~\ref{lem:Z1} (for $j \le k-1$) alternately now, it follows that 
$T_j \subseteq L_3$ for $j \le k-1$. Thus $\cup_{i=1}^{k-1}T_i \subseteq L_3$. We will now show that $T_k$ is a subset of~$L_3$.

\smallskip

Let $H_k$ denote the graph $H$ in step~4 in the $k$-th iteration of the While-loop in our 
algorithm. This is the graph where we determine the even posts that will get demoted from 
$Y$ to~$Z$. In step~4 of the $k$-th iteration of the While-loop, the set $X = F \setminus \cup_{i=2}^k F_i$ 
(call this set $X_k$), $Z = \cup_{i=1}^{k-1}T_i$ (call this set $Z_k$), and let $Y_k$ be the set of posts 
outside~$X_k \cup Z_k$. The edge set of $H_k$ is as follows: 
\begin{itemize}
\item for each $a \in A$: if the rank of $a$'s most preferred post $b$ in $Y_k$ is $\le r_a$, then the 
edge $(a,b)$ belongs to $H_k$
\item for $a \in A \setminus\Nbr(Z_k)$: the edge~$(a,f(a))$ is also present in~$H_k$.
\end{itemize}

Let us refer to posts in $S$ 
as {\em $s$-posts}: recall that these are odd/unreachable non-$f$-posts in the graph~$H_1$.
We will now show that all $s$-posts in $L_2$ are odd/unreachable in $H_k$;
so every $s$-post that is even in $H_k$ has to be in $L_3$, in other words, $T_k \subseteq L_3$. 
Let $G'_0$ be the subgraph of $G'$ with the set of posts restricted to $L_1 \cup L_2$ 
(see Fig.~\ref{fig:last-claim}).
Consider the subgraph $G'_k$ of $G'_0$ obtained by deleting edges missing in $H_k$ from~$G'_0$.

We now show that $G'_k$ contains all edges in $G'_0$ incident on $s$-posts in~$L_2$. 
This is because any edge $(a,b)$ incident on an $s$-post $b \in L_2$ in $G'_0$ is present 
in $H_k$ also. Since the edge $(a,b)$ belongs to $G'_0$, the post $b$ has to be ranked $r_a$ 
in $a$'s preference list and there is no $f$-post in $L_2$ of rank better than $r_a$ in $a$'s 
list. If the edge $(a,b)$ does not exist in $H_k$, then it means there is some $f$-post in 
$Y_k$ that $a$ prefers to~$b$. All $f$-posts in $Y_k$ are in $\cup_{i=2}^k F_i$ and we are given 
that $\cup_{i=2}^k F_i \subseteq L_2$. Since we know there is no $f$-post in $L_2$ that $a$ 
prefers to $b$, it follows that $b$ has to be $a$'s most preferred post in $Y_k$ and so the 
edge $(a,b)$ belongs to~$H_k$. Thus $G'_k$, whose edge set is the intersection of the edge sets 
of $G'_0$ and $H_k$, contains all edges in $G'_0$ incident on $s$-posts in~$L_2$. 

\begin{figure}[h]
\centerline{\resizebox{0.27\textwidth}{!}{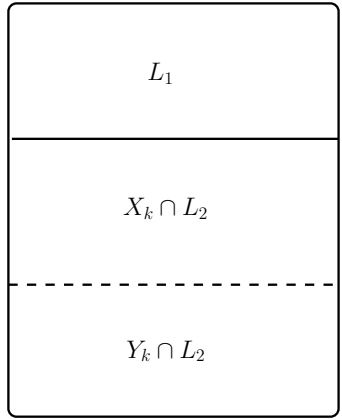}}
\caption{The set of posts in $G'_0$ can be viewed as $L_1 \cup (X_k\cap L_2) \cup (Y_k\cap L_2)$. 
All $s$-posts in $L_2$ are in $Y_k\cap L_2$.}
\label{fig:last-claim}
\end{figure}

Every post in $L_1 \cup L_2$ is odd/unreachable in $G'_0$ since the matching $M^*$ restricted 
to the edge set of $G'_0$ is $(L_1 \cup L_2)$-complete. We have shown that $G'_k$ contains all 
edges in $G'_0$ incident on $s$-posts in $L_2$: thus all $s$-posts in $L_2$ 
are odd/unreachable in~$G'_k$. It is easy to see that all {\em top-ranked} edges in 
$G'_0$ incident on $f$-posts in $Y_k \cap L_2$ are also present in $G'_k$: each such post has 
a degree~1 neighbor in $G'_k$, thus all $f$-posts in $Y_k \cap L_2$ are also odd/unreachable 
in~$G'_k$.

We now claim that all posts in $L_1$ are also odd/unreachable in~$G'_k$. We first show that 
all edges incident on $L_1$ in $G'_0$ are present in~$H_k$. This is because each edge $(a,b)$ 
in $G'_0$ such that $b \in L_1$ is incident on an applicant $a \in A\setminus\Nbr(L_3)$ such 
that $b = f(a)$ and we know the graph $H_k$ has $(a,f(a))$ edges for all 
$a \in A\setminus\Nbr(Z_k) \supseteq A\setminus\Nbr(L_3)$ since 
$Z_k = \cup_{i=1}^{k-1}T_i \subseteq L_3$.

In $G'_k$, each vertex $b \in L_1$ either has a degree~1 neighbor (in which case our claim is 
true) or it has a degree~2 neighbor $a$ whose other neighbor is in $Y_k \cap L_2$, i.e.
it is not in $X_k \cap L_2$. This is because $a$ cannot have 2 neighbors in $X_k$ in the graph 
$H_k$ and we know $L_1 \subseteq X_k$ since all $f$-posts missing in $X_k$ (these are posts in 
$\cup_{i=2}^{k-1}F_i$) are absent from $L_1$ also. Since all posts in $Y_k \cap L_2$ 
are odd/unreachable in~$G'_k$, it follows that all posts in $L_1$ are also odd/unreachable 
in~$G'_k$.

\smallskip

Let us now compare the graph $H_k$ with the graph~$G'_k$. The graph $H_k$ has additional 
vertices: these are the ones in $Y_k \cap L_3$ and the new edges in $H_k$ (new relative to $G'_k$)
 belong to the following two classes: (i)~$\Nbr(L_3)\times(Y_k \cap L_3)$
and (ii)~$A \times (L_1 \cup (Y_k \cap L_2))$. This is because every edge incident on 
$X_k \cap L_2$ in $H_k$ (these are all top-ranked edges) is present in $G'_0$ as well. 

Consider any new edge $(a,b)$ in $H_k$ of type~(i), 
i.e.\ $(a,b) \in \Nbr(L_3)\times(Y_k \cap L_3)$. Since $(a,b)$ belongs to $H_k$, it must be 
the case that $a$'s most preferred neighbor in $Y_k$ is~$b$. So the post $b$ is
ranked $r_a$ in $a$'s list and $a$ has no neighbor of rank better than $r_a$ in~$Y_k$. 
Recall that $G'_0$ has no edge in  $\Nbr(L_3)\times L_1$. So the only edge that can be incident 
on $a$ in the graph $G'_k$ is an edge to $f(a)$ in $X_k \cap L_2$. 

Consider any connected component $C$ in $G'_k$ that contains an $s$-post in $L_2$: every post 
here belongs to either $L_1$ or $Y_k \cap L_2$, in other words, there is no post in 
$X_k \cap L_2$ here. This is because there is no applicant $a$ in $G'_k$ with neighbors in 
$Y_k \cap L_2$ and  $X_k \cap L_2$ as this means $a$ has two neighbors in $L_2$, which  is 
forbidden in~$G'_0$. 
Similarly, there is no applicant $a'$ in $G'_k$ with neighbors in $L_1$ and 
$X_k \cap L_2$ as this means $a$ has two neighbors in $X_k$, which  is forbidden in~$H_k$. 
Thus $C$ has no post from $X_k \cap L_2$.

So the new edges in $H_k$ of type~(i) do not touch components in $G'_k$ that contain $s$-posts 
in~$L_2$. All the new edges incident upon these components have their endpoints in 
$L_1 \cup (Y_k \cap L_2)$. These posts  are already odd/unreachable in~$G'_k$. So these posts 
remain  odd/unreachable in~$H_k$. Hence every $s$-post in $L_2$ is odd/unreachable in~$H_k$. \qed

\subsubsection*{The augmented graph $G'$.}
The matching $M^*$ need not be $A$-complete. However it would help us to assume that $M^*$ is 
$A$-complete, so we augment $M^*$ by adding $(a,\ell(a))$ edges for every $a \in A$ that
is unmatched in~$M^*$. Recall that $\ell(a)$ is the dummy last resort post of~$a$.
However the augmented matching $M^*$ need not belong to the graph $G'$ any longer -- hence
we augment $G'$ also by adding some dummy vertices and some edges as described below. 

The augmentation of $G'$ is analogous to phase~(III) of our algorithm -- we augment $G'$ as
follows: let $L_2 = L_2 \cup D$, where $D = \{\ell(a): a \in A\ \mathrm{and}\ r_a=\infty\}$;
if $\Nbr(\{a\}) \subseteq L_1$, then add $(a,\ell(a))$ to~$G'$. Thus when compared to $G'$, 
the augmented $G'$ has some new vertices (all these are dummy last resort posts) and 
some new edges -- each new edge is of the form $(a,\ell(a))$ where $\ell(a)$ is $a$'s only
neighbor in $L_2\cup L_3$.
These new edges are enough to show the following lemma.

\begin{lemma}
\label{lem:augmented}
The augmented matching $M^*$ belongs to the augmented graph~$G'$.
\end{lemma}
\begin{proof}
Before the augmentations of $G'$ and $M^*$, the matching $M^*$ belonged to the graph~$G'$ (by Lemma~\ref{lem:main}). We now need to 
show that if $a$ is left unmatched in $M^*$ (before augmentation), then $r_a = \infty$ and all of $a$'s neighbors belong to~$L_1$.

Suppose $a$ is left unmatched in $M^*$ and $r_a < \infty$. 
Since $r_a < \infty$, there is a post $b \notin F$ such that the post $b$ has rank $r_a$ in 
$a$'s preference list. Consider the alternating path $p = a$-$b$-$a'$-$f(a')$-$a''$, where 
$a' = M^*(b)$ and $a'' = M^*(f(a'))$. The matching $M^* \oplus p$ matches $a$ to $b$ and promotes
$a'$ to its top post $f(a')$ and leaves $a''$ unmatched.
Thus $M^* \oplus p$ is more popular than $M^*$, a contradiction.

So let us assume $r_a = \infty$ and $a$ was left unmatched in~$M^*$. Suppose $a$ has some neighbor $b_0$ outside~$L_1$. 
The post $b_0$ has to be in $F$ because $r_a = \infty$, i.e.\ $a$ has no neighbors outside~$F$. 
Since $F \subseteq L_1 \cup L_2$ (by Lemma~\ref{lemma-F}.1), it follows that $b_0 \in L_2$. Let $a_0 = M^*(b_0)$; if 
$b_0 \ne f(a_0)$, then we again have an alternating path $p = a$-$b_0$-$a_0$-$f(a_0)$-$a_1$, 
where $a_1 = M^*(f(a_0))$ such that $M^* \oplus p$ is more popular than~$M^*$.
This contradicts the popularity of~$M^*$.

So suppose  $b_0 = f(a_0)$ and $b_0 \in L_2$ because $a_0 \in \Nbr(L_3)$. We know from 
Claim~\ref{lem:alt-path} that there is a desired alternating path $\rho_{a_0}$, where
either the last two non-matching edges are labeled~$+1$ or the last post in $\rho_{a_0}$ 
is unmatched. 
Consider the alternating path $\rho$ which is the path $a$-$b_0$-$a_0$ followed by the path~$\rho_{a_0}$. 
It is easy to see that $M^* \oplus \rho$ is more popular than $M^*$, 
a contradiction to the popularity of~$M^*$. \qed
\end{proof}

Since the augmented $M^*$ is an $A$-complete matching, it follows from Lemma~\ref{lem:augmented} 
that the augmented graph $G'$ admits an $A$-complete matching. 
Theorem~\ref{last-lemma} uses Lemma~\ref{lem:L1andX} to show that if the augmented graph $G'$ 
admits an $A$-complete matching, then so does the graph $H$ constructed by our algorithm. 

\begin{theorem}
\label{last-lemma}
If $H$ does not admit an $A$-complete matching, then the augmented graph $G'$ cannot admit an 
$A$-complete matching.
\end{theorem}
\begin{proof}
We will use $G'$ to refer to the {\em augmented} graph $G'$ in this proof.
The rules for adding edges in $H$ and in $G'$ are exactly the same -- the only difference is 
in the partition $\langle X, Y, Z\rangle$ on which $H$ is based vs the partition 
$\langle L_1, L_2, L_3\rangle$ on which $G'$ is based. If $\langle X, Y, Z\rangle = 
\langle L_1, L_2, L_3\rangle$, then the graphs $H$ and $G'$ are exactly the same. 

\begin{figure}[h]
\centerline{\resizebox{0.65\textwidth}{!}{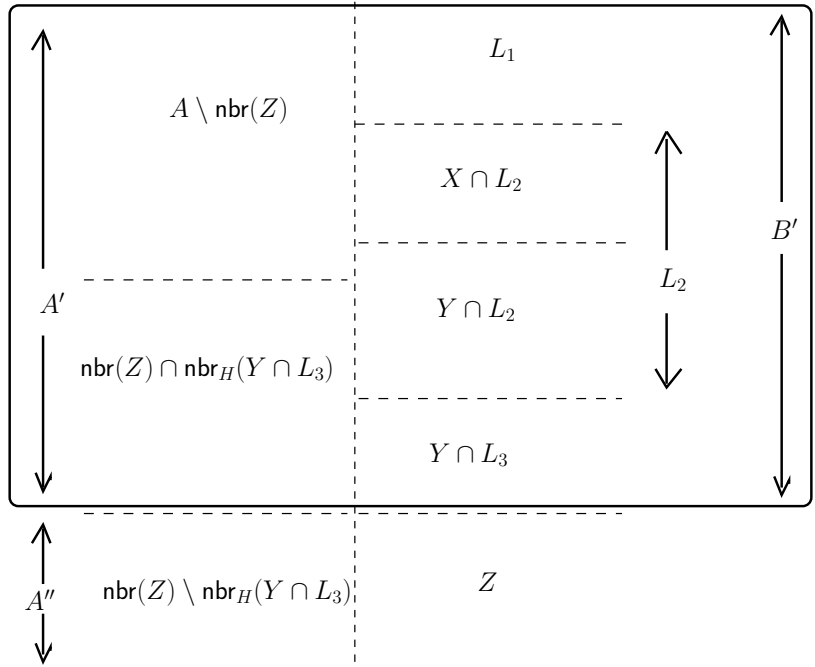}}
\caption{The part of $G'$ inside the box will be called~$G''$. 
The graph $G'$ has no edge between any applicant in $A'$ and any post in~$Z$.}
\label{fig:last-thm}
\end{figure}

Fig.~\ref{fig:last-thm} denotes how the partition $\langle X,Y,Z\rangle$ can be
modified to the partition $\langle L_1,L_2,L_3\rangle$. We know from  Lemma~\ref{lem:L1andX} 
that $X \supseteq L_1$ and $Z \subseteq L_3$.
Consider the subgraph $G''$ of $G'$ induced on the vertex set 
$A' = (A \setminus \Nbr(Z)) \cup (\Nbr(Z)\cap\Nbr_H(Y \cap L_3))$ and $B' = X \cup Y$. 
This is the part bounded by the box in Fig.~\ref{fig:last-thm}. 
In our analysis, we can essentially separate $G'$ into $G''$ and the part outside $G''$ 
due to the following claim that says $G'$ has no edges between $A'$ and~$Z$.

\begin{new-claim}
\label{claim1}
$G'$ has no edge $(a,b)$ where $a \in A'$ and $b \in Z$. 
\end{new-claim}
\begin{proof}
Any applicant $a \in A'$ has to belong to either $A \setminus \Nbr(Z)$ or to 
$\Nbr(Z)\cap\Nbr_H(Y\cap L_3)$ (see Fig.~\ref{fig:last-thm}). 
There is obviously no edge in $G$ between a vertex in $A \setminus \Nbr(Z)$ and any vertex in~$Z$. So suppose $a \in \Nbr(Z)\cap\Nbr_H(Y \cap L_3)$. For $b \in L_3$, if the edge $(a,b)$
is in $G'$, then $b$ has to be $a$'s most preferred post in~$L_3$. We will now show that 
$b \in Y \cap L_3$, equivalently $b \notin Z$. 
Thus $G'$ has no edge $(a,b)$ where $a \in A'$ and $b \in Z$.

Since $a\in \Nbr_H(Y \cap L_3)$, the graph $H$ contains an edge between $a$ and some 
$b' \in  Y \cap L_3$. Recall that an element of $Y \cap L_3$ is a real post in~$Y$. 
By the rules of including edges in $H$, it follows that the rank of $b'$ in $a$'s preference 
list is $\le r_a$. The entire set $L_3$ cannot contain any post of rank better than $r_a$ for any 
$a \in A$ since any post of rank better than $r_a$ in $a$'s list belongs to $F$ while 
$L_3 \cap F = \emptyset$ (by Lemma~\ref{lemma-F}.1). So $b'$ has rank $r_a$ in $a$'s list.
Thus $a$'s most preferred neighbor in $L_3$ belongs to $Y \cap L_3$. \qed
\end{proof}

Let $G_0$ be the subgraph of $G''$ obtained by deleting from $G''$ the edges that are absent 
in~$H$. Thus $G_0$ is a subgraph of both $G'$ and~$H$. The following claim (whose proof is given 
after the proof of Theorem~\ref{last-lemma}) will be useful to us.

\begin{new-claim}
\label{claim2}
All posts in $(X \cap L_2) \cup (Y \cap L_3)$ are odd/unreachable in~$G_0$.
Moreover, every edge $(a,b)$ in $G'$ that is missing in $H$ satisfies 
$b \in (X \cap L_2) \cup (Y \cap L_3)$.
\end{new-claim}

Consider the graph $G_1$ whose edge set is the intersection of the edge sets of $G'$
and~$H$. Equivalently, $G_1$ can be constructed by adding to the edge set of $G_0$,
the edges incident on $A'' = \Nbr(Z)\setminus\Nbr_H(Y \cap L_3)$ that are present 
in both $G'$ and $H$ (see Fig.~\ref{fig:last-thm}).
This is due to the fact that $G'$ has no edge in $A' \times Z$.  

We claim that all posts in $(X \cap L_2) \cup (Y \cap L_3)$ are odd/unreachable in~$G_1$. This is because Claim~\ref{claim2} tells us that each post in this set is 
odd/unreachable in $G_0$ and due to the absence of $A' \times Z$ edges in $G'$, the graph 
$G_1$ has no {\em new} edge (new when compared to $G_0$) incident on the set $A'$ of 
applicants in~$G_0$. Hence all posts in $(X \cap L_2) \cup (Y \cap L_3)$ remain 
odd/unreachable in~$G_1$.

Claim~\ref{claim2} also tells us that all edges in $G'$ that are missing in $H$ are incident 
on posts in $(X \cap L_2) \cup (Y \cap L_3)$. We know that all these posts are 
odd/unreachable in $G_1$, hence $G'$ has no {\em new} edge (new when compared to $G_1$) on 
posts that are {\em even} in~$G_1$. Thus the size of a maximum matching in $G'$ equals the 
size of a maximum matching in~$G_1$. This is at most the size of a maximum matching in $H$, 
since $G_1$ is a subgraph of~$H$. 
Hence if $H$ has no $A$-complete matching, then neither does~$G'$. \qed
\end{proof}

\paragraph{Proof of Claim~\ref{claim2}.}
We will now show that all posts in $(X \cap L_2) \cup (Y \cap L_3)$ are odd/unreachable in~$G_0$.
Let $a$ be an applicant with degree~2 in the graph $G_0$, let $b_1$ and $b_2$ be the two 
neighbors of $a$, where $b_1$ is the more preferred neighbor of~$a$. We claim either 
(i)~$b_1 \in X \cap L_2$ and $b_2 \in Y \cap L_3$ or
(ii)~$b_1 \in L_1$ and $b_2 \in Y \cap L_2$.
This is because of the following.

\begin{itemize}
\item There is no applicant in $G_0$ with edges to both a post in $L_1$ and post in 
$X \cap L_2$. If there was such an applicant $a$, then $a$ would have two
neighbors in the set~$X$, which is forbidden in~$H$. Recall that
any edge in $G_0$ is an edge in $H$ as well. 
\item There is no applicant in $G_0$ with edges to both a post in $X \cap L_2$ and a post in 
$Y \cap L_2$. If there was such an applicant $a$, then $a$ would have two
neighbors in the set~$L_2$, which is forbidden in~$G'$. Recall that
any edge in $G_0$ is an edge in $G'$ as well. 
\item There is no applicant in $G_0$ with edges to both a post in $Y \cap L_2$ and a post in 
$Y \cap L_3$. If there was such an applicant $a$, then $a$ would have two
neighbors in the set~$Y$, which is forbidden in~$H$.
\item There is no applicant $a$ in $G_0$ with edges to both a post in $L_1$ and a post in 
$Y \cap L_3$. This is because $G'$ cannot contain such a pair of edges as it is only applicants 
in $A \setminus \Nbr(L_3)$ that are adjacent to posts in~$L_1$.
\end{itemize}

Thus in the graph $G_0$, vertices in  $(X \cap L_2) \cup (Y \cap L_3)$ and those in 
$L_1 \cup (Y \cap L_2)$ belong to different connected components. 
Note that all dummy posts belong to $Y \cap L_2$. So none of these posts belongs to any
connected component in $G_0$ that contains vertices in  $(X \cap L_2) \cup (Y \cap L_3)$.
Consider the subgraph $H'$ of $H$,  obtained by restricting the set of posts in $H$ to 
{\em real} posts in $X \cup Y$. All real posts in $X \cup Y$ are odd/unreachable in~$H'$.  
Since $(X \cap L_2) \cup (Y \cap L_3)$ consists of real posts, all these posts are 
odd/unreachable in $H'$.

\smallskip

We now claim that all posts in $(X \cap L_2) \cup (Y \cap L_3)$ remain odd/unreachable in~$G_0$. 
In the first place, every edge $(a_0,b_0)$ in $H'$ incident on any vertex $b_0 \in Y \cap L_3$ 
is present in $G''$ as well. This is because  $a_0 \in A'$ and if $b_0 \in Y \cap L_3$ is the 
most preferred post in $Y$ for applicant $a_0$, then the rank of $b_0$ in $a_0$'s preference 
list is $r_{a_0}$ and thus $b_0$ is also $a_0$'s most preferred post in $L_3$, so the edge 
$(a_0,b_0)$ belongs to the graph~$G''$. Similarly every edge $(a_1,b_1)$ in $H'$ incident on 
any post $b_1 \in X \cap L_2$ is present in $G''$ as well -- this is  because $a_1 \in A'$ 
and $b_1$ has to be $f(a_1)$ for the edge $(a_1,b_1)$ to exist in~$H'$. Thus $b_1$ is also 
$a_1$'s most preferred post in~$L_2$. Hence all edges in $H'$ incident on posts in 
$(X \cap L_2) \cup (Y \cap L_3)$ are present in~$G_0$.

Let $b$ be any post in $(X \cap L_2) \cup (Y \cap L_3)$. The connected component in $G_0$ that 
contains $b$ can be obtained by taking the connected component containing $b$ in $H'$ and 
deleting all vertices in $L_1 \cup (Y \cap L_2)$ from this component. Since no edge incident on 
$b$ has been deleted here and because $b$ is odd/unreachable in the starting component, it 
follows that $b$ is odd/unreachable in~$G_0$.

\smallskip

We will now show the second part of Claim~\ref{claim2}: every edge $(a,b)$ in $G'$ that is 
missing in $H$ satisfies $b \in (X \cap L_2) \cup (Y \cap L_3)$.
We partitioned the set $B$ of posts into five sets (refer to Fig.~\ref{fig:last-thm}). These are 
$L_1$, $X \cap L_2$, $Y \cap L_2$, $L_3 \cap Z$, and~$Z$. We will now show that every 
edge in $G'$ that is incident on $L_1 \cup (Y \cap L_2) \cup Z$ is present in $H$ also.
\begin{itemize}
\item Any edge $(a,b)$ in $G'$ where $b \in L_1$ is such that $f(a) = b$ and $a \in A \setminus \Nbr(L_3)$.
Since $L_3 \supseteq Z$ (by Lemma~\ref{lem:L1andX}), this means $a \in  A \setminus \Nbr(Z)$. Thus $H$ also contains the edge~$(a,b)$.
\item Any edge $(a,b)$ in $G'$ where $b \in Y \cap L_2$ is such that $b$ is $a$'s most preferred post in $L_2$ and
the rank of $b$ in $a$'s preference list is $\le r_a$. Note that $Y \setminus L_2 = (Y \cap L_3) \subseteq B\setminus F$ 
(by Lemma~\ref{lemma-F}.1). Thus the rank of $a$'s most preferred post in $Y \setminus L_2$ is $\ge r_a$ and hence no post in 
$Y \setminus L_2$ can be preferred to~$b$. So the post $b$ is, in fact, $a$'s most preferred post in~$Y$. Thus the edge $(a,b)$ 
belongs to $H$ as well. 
\item Any edge $(a,b)$ in $G'$ where $b \in Z$ is such that $b$ is $a$'s most preferred post in~$L_3$.
Since $L_3 \supseteq Z$, this means $b$ is $a$'s most preferred post in~$Z$. Thus $H$ also contains the edge~$(a,b)$. 
\end{itemize}
Thus every edge $(a,b)$ in $G'$ that is missing in $H$ satisfies 
$b \in (X \cap L_2) \cup (Y \cap L_3)$. \qed

\medskip

Theorem~\ref{last-lemma}, along with Lemma~\ref{lem:augmented}, finishes the proof of the 
necessary part of Theorem~\ref{thm:correctness} and this completes 
the proof of correctness of our algorithm. We now analyze its running time.

Observe that we can maintain the most preferred posts in $X, Y$, and $Z$ for all 
applicants over all iterations in $O(m)$ time, where $m$ is the number of edges in $G$. 
To begin with, the most preferred non-$f$-post for all applicants can be determined in $O(m)$ time. Thereafter, whenever a post $b$ moves from $X$ to $Y$ (similarly, from $Y$ to $Z$), we charge $b$ a cost of the degree of $b$ to pay for checking if any of its neighbors now has $b$ as its most preferred post in $Y$ (resp.,~$Z$). 

Let $n$ be the number of vertices in~$G$. The number of iterations is $O(n)$ and 
the most expensive step in each iteration is finding a maximum matching in a subgraph where 
each vertex in $A$ has degree at most~2. It is easy to see that this step can
be performed in $O(n)$ time. Thus the running time of our algorithm is~$O(n^2)$.
Hence we can conclude Theorem~\ref{thm:popular} stated in Section~\ref{sec:intro}.

\smallskip

There are instances on $O(n)$ vertices and $O(n)$ edges where our algorithm takes $\Theta(n^2)$ time. 
Consider the example given in Fig.~\ref{fig:tight}: here there are $2n+1$ applicants and $2n+2$ posts and the
number of edges is~$5n+2$.
For each $1 \le i \le n$, we have $f(a_i) = f(a'_i) = f_i$ and $s_i$ is the most preferred non-$f$-post for both $a_i$ and~$a'_i$.
For $a_0$, we have $f(a_0) = f_0$ and $a_0$'s most preferred non-$f$-post is~$s_0$.

\begin{figure}[h]
\centerline{\resizebox{0.6\textwidth}{!}{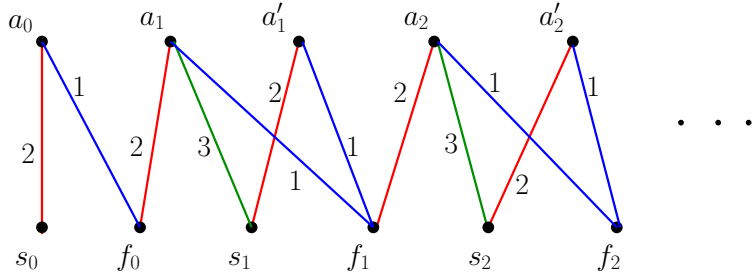}}
\caption{The preferences of applicants are indicated on the edges. Our algorithm runs for $n+1$ iterations here.}
\label{fig:tight}
\end{figure}

In the starting graph $H_1$, there is exactly one post that is even in $Y_1$: this is $s_0$ and 
so $s_0$ moves from $Y_1$ to~$Z_1$. In the second iteration, $f_0$ has no applicant in $A \setminus \Nbr(Z)$ that regards it
as a top post  and this causes the demotion of $f_0$ from $X_2$ to~$Y_2$. 
Now the post $f_0$ is the most preferred post in $Y_2$ for $a_1$ and this makes
$s_1$ even in $Y_2$ and causes $s_1$ to move from $Y_2$ to~$Z_2$. 

This makes both $a_1$ and $a'_1$ belong to $\Nbr(Z)$ and hence $f_1$ gets isolated in $H_3$ and so $f_1$ moves from $X_3$ to~$Y_3$. Now $f_1$ becomes the most preferred post in $Y_3$ for $a_2$ and this causes $s_2$ to move from $Y_3$ to $Z_3$ and so on. 
Thus our algorithm runs for $n+1$ iterations. This instance admits popular matchings, for instance, 
$\{(a_0,f_0),(a_i,f_i),(a'_i,s_i):\, 1\le i \le n\}$ is a popular matching here.

\section{Our $\mathsf{NP}$-hardness result}
\label{sec:hardness}
Given a matching $M$ in $G = (A \cup B,E)$, it was shown in~\cite{BIM09} that $M$ can be tested for popularity in 
$O(\sqrt{|V|}\cdot|E|)$ time (even in the presence of ties), where $|V| = |A \cup B|$. We now show the $\mathsf{NP}$-hardness 
of the 2-sided popular matching problem in $G$ with 1-sided ties using the \textsc{(2,2)-e3-sat} problem.  

Recall that the \textsc{(2,2)-e3-sat} problem takes as its input a Boolean formula $\mathcal{I}$ in CNF, where each clause contains 
three literals and every variable appears exactly twice in unnegated form and exactly twice in negated form in the clauses. The problem 
is to determine if $\mathcal{I}$ is satisfiable or not and this problem is $\mathsf{NP}$-complete~\cite{BKS03}.

\paragraph{Constructing a popular matching instance $G = (A \cup B,E)$ from~$\mathcal{I}$.}
Let $\mathcal{I}$ have $m$ clauses and $n$ variables. The instance $G$ constructed consists of $n$ variable gadgets, $m$ clause gadgets, and some interconnecting edges between these, see Fig.~\ref{fi:pop_np}. A \emph{variable gadget} representing variable $v_j$, for $1\leq j\leq n$, is a 4-cycle on vertices 
$a_{j_1},  b_{j_1}, a_{j_2}$, and~$b_{j_2}$, where $a_{j_1},a_{j_2}\in A$ and $b_{j_1},b_{j_2}\in B$. 
A \emph{clause gadget} representing clause $C_i$, for $1 \le i \leq m$, is a subdivision graph of a claw. Its 
edges are divided into three classes: $c_i \in B$ is at the center, the neighbors of $c_i$ are $x_{i_1}, x_{i_2}, x_{i_3} \in A$, and 
finally, each of $x_{i_1}, x_{i_2}, x_{i_3}$ is adjacent to its respective copy in~$\mathcal{Y}_i = \{y_{i_1}, y_{i_2},y_{i_3}\}$, where 
$\mathcal{Y}_i \subseteq B$.

A vertex in $\mathcal{Y}_i$ represents an appearance of a variable. For instance, $y_{3_1}$ is the first literal of the third clause. Each of 
the vertices in $\mathcal{Y}_i$ is connected to a vertex in the variable gadget via an \emph{interconnecting edge}. Vertex $y_{i_k}$ is 
connected to the gadget standing for variable $j$ if the $k$-th literal of the $i$th clause is either $v_j$ or~$\neg v_j$. If it is $v_j$,
then the interconnecting edge ends at $a_{j_1}$, else at~$a_{j_2}$. The preferences of this instance can be seen in 
Fig.~\ref{fi:pop_np}. The constructed graph trivially satisfies both conditions claimed in Section~\ref{sec:intro}, i.e.\ 
every vertex in $A$ has a strict preference list of length 2 or 4 and every vertex in $B$ has either a strict preference list of length~2 
or a single tie of length~2 or 3 as a preference list.
\tikzstyle{vertex} = [circle, draw=black, fill=black, inner sep=0pt,  minimum size=5pt]
\tikzstyle{edgelabel} = [circle, fill=white, inner sep=0pt,  minimum size=15pt]
\begin{center}
\begin{figure}[h!]
\centering
	\pgfmathsetmacro{\d}{3}
	\pgfmathsetmacro{\b}{5}
	\pgfmathsetmacro{\k}{10}
	\pgfmathsetmacro{\c}{2}
\begin{minipage}{0.45\textwidth}
\centering
\begin{tikzpicture}[scale=0.85, transform shape]

	\node[vertex, label=below:$c_{i_1}$] (c1) at ($(0,0) + (-\k, 1)$) {};
	\node[vertex, label=below:$x_{i_1}$] (x11) at ($(c1) + (\d, \c)$) {};
	\node[vertex, label=below:$x_{i_2}$] (x12) at ($(c1) + (\d, 0)$) {};
	\node[vertex, label=below:$x_{i_3}$] (x13) at ($(c1) + (\d, -\c)$) {};
	\node[vertex, label=below:$y_{i_1}$] (y11) at ($(x11) + (\d, 0)$) {};
	\node[vertex, label=below:$y_{i_2}$] (y12) at ($(x12) + (\d, 0)$) {};
	\node[vertex, label=below:$y_{i_3}$] (y13) at ($(x13) + (\d, 0)$) {};
	
	\draw [] (c1) -- node[edgelabel, near start] {1} node[edgelabel, near end] {1} (x11);
	\draw [] (c1) -- node[edgelabel, near start] {1} node[edgelabel, near end] {1} (x12);
	\draw [] (c1) -- node[edgelabel, near start] {1} node[edgelabel, near end] {1} (x13);
	
	\draw [] (x11) -- node[edgelabel, near start] {2} node[edgelabel, near end] {1} (y11);
	\draw [] (x12) -- node[edgelabel, near start] {2} node[edgelabel, near end] {1} (y12);
	\draw [] (x13) -- node[edgelabel, near start] {2} node[edgelabel, near end] {1} (y13);
	
	\draw [dotted] (y11) -- node[edgelabel] {2} ($(y11)+(1,1)$);
	\draw [dotted] (y12) -- node[edgelabel] {2} ($(y12)+(1,1)$);
	\draw [dotted] (y13) -- node[edgelabel] {2} ($(y13)+(1,1)$);

\end{tikzpicture}

\end{minipage}\hspace{1.22cm}\begin{minipage}{0.45\textwidth}
\centering
	\begin{tikzpicture}[scale=0.9, transform shape]
	\node[vertex, label=below:$a_{j_1}$] (A1) at (0, 0) {};
	\node[vertex, label=below:$b_{j_1}$] (B1) at (\d+1, 0) {};
	\node[vertex, label=below:$a_{j_2}$] (A2) at (0, -\d-1) {};
	\node[vertex, label=below:$b_{j_2}$] (B2) at (\d+1, -\d-1) {};
	
	\draw [] (A1) -- node[edgelabel, near start] {1} node[edgelabel, near end] {1} (B1);
	\draw [] (A1) -- node[edgelabel, near start] {4} node[edgelabel, near end] {1} (B2);
	\draw [] (A2) -- node[edgelabel, near start] {1} node[edgelabel, near end] {1} (B1);
	\draw [] (A2) -- node[edgelabel, near start] {4} node[edgelabel, near end] {1} (B2);
	
	\draw [dotted] (A1) -- node[edgelabel] {2} ($(A1)+(-1,1)$);
	\draw [dotted] (A1) -- node[edgelabel] {3} ($(A1)+(-1,-1)$);
	\draw [dotted] (A2) -- node[edgelabel] {2} ($(A2)+(-1,1)$);
	\draw [dotted] (A2) -- node[edgelabel] {3} ($(A2)+(-1,-1)$);
	
	\end{tikzpicture}
\end{minipage}

\vspace{1.5cm}

	\begin{tikzpicture}[scale=0.9, transform shape]
	
	\node[vertex] (A1) at (0, 0) {};
	\node[vertex] (B1) at (\d, 0) {};
	\node[vertex] (A2) at (0, -\d) {};
	\node[vertex] (B2) at (\d, -\d) {};
	
	\draw [ultra thick, red] (A1) -- (B1);
	\draw [] (A1) -- node [near start, fill=white]{(-1,0)} (B2) ;
	\draw [] (A2) -- node [near start, fill=white]{(+1,0)} (B1);
	\draw [ultra thick, red] (A2) -- (B2);
	\draw [dotted] (A1) -- ($(A1)+(-1,-0.5)$);
	\draw [dotted] (A2) -- ($(A2)+(-1,0.5)$);

	\node[vertex] (c2) at ($(A1) + (-\k, 0)$) {};
	\node[vertex] (x21) at ($(c2) + (\d/1.5, \c/1.5)$) {};
	\node[vertex] (x22) at ($(c2) + (\d/1.5, 0)$) {};
	\node[vertex] (x23) at ($(c2) + (\d/1.5, -\c/1.5)$) {};
	\node[vertex] (y21) at ($(x21) + (\d/1.5, 0)$) {};
	\node[vertex] (y22) at ($(x22) + (\d/1.5, 0)$) {};
	\node[vertex] (y23) at ($(x23) + (\d/1.5, 0)$) {};
	
	\draw [] (c2) -- node [above, fill=white]{(+1,0)} (x21);
	\draw [ultra thick, red] (c2) -- (x22);
	\draw [] (c2) -- node [below, fill=white]{(+1,0)} (x23);
	
	\draw [ultra thick, red] (x21) -- (y21);
	\draw []  (x22) -- node [above, fill=white]{(-1,+1)} (y22);
	\draw [ultra thick, red] (x23) -- (y23);
	
	\draw [dotted] (y21) -- ($(y21)+(1,-0.5)$);
	\draw [dotted] (y22) -- node [above, fill=white]{(+1,-1)}(A1);
	\draw [dotted] (y23) --($(y23)+(1,0.5)$);
	
	\node[vertex] (c1) at ($(A1) + (-\k, -4)$) {};
	\node[vertex] (x11) at ($(c1) + (\d/1.5, \c/1.5)$) {};
	\node[vertex] (x12) at ($(c1) + (\d/1.5, 0)$) {};
	\node[vertex] (x13) at ($(c1) + (\d/1.5, -\c/1.5)$) {};
	\node[vertex] (y11) at ($(x11) + (\d/1.5, 0)$) {};
	\node[vertex] (y12) at ($(x12) + (\d/1.5, 0)$) {};
	\node[vertex] (y13) at ($(x13) + (\d/1.5, 0)$) {};
	
	\draw [ultra thick, red] (c1) -- (x11);
	\draw [] (c1) -- node [below=3pt, near end, fill=white]{(+1,0)} (x12);
	\draw [] (c1) -- node [below, fill=white]{(+1,0)} (x13);
	
	\draw [] (x11) --  node [above, fill=white]{(-1,+1)} (y11);
	\draw [ultra thick, red] (x12) -- (y12);
	\draw [ultra thick, red] (x13) -- (y13);
	
	\draw [dotted] (y11) -- ($(y11)+(1,0.5)$);
	\draw [dotted] (y12) -- ($(y12)+(1,0.5)$);
	\draw [dotted] (y13) --  node [above, fill=white]{(-1,+1)} (A2);
\end{tikzpicture}

\caption{A clause gadget, a variable gadget, and the structure of the entire construction with a variable that appears at the second place in the first clause in unnegated form and at the third place in the second clause in negated form. The thick red matching corresponds to a true variable.}
\label{fi:pop_np}
\end{figure}
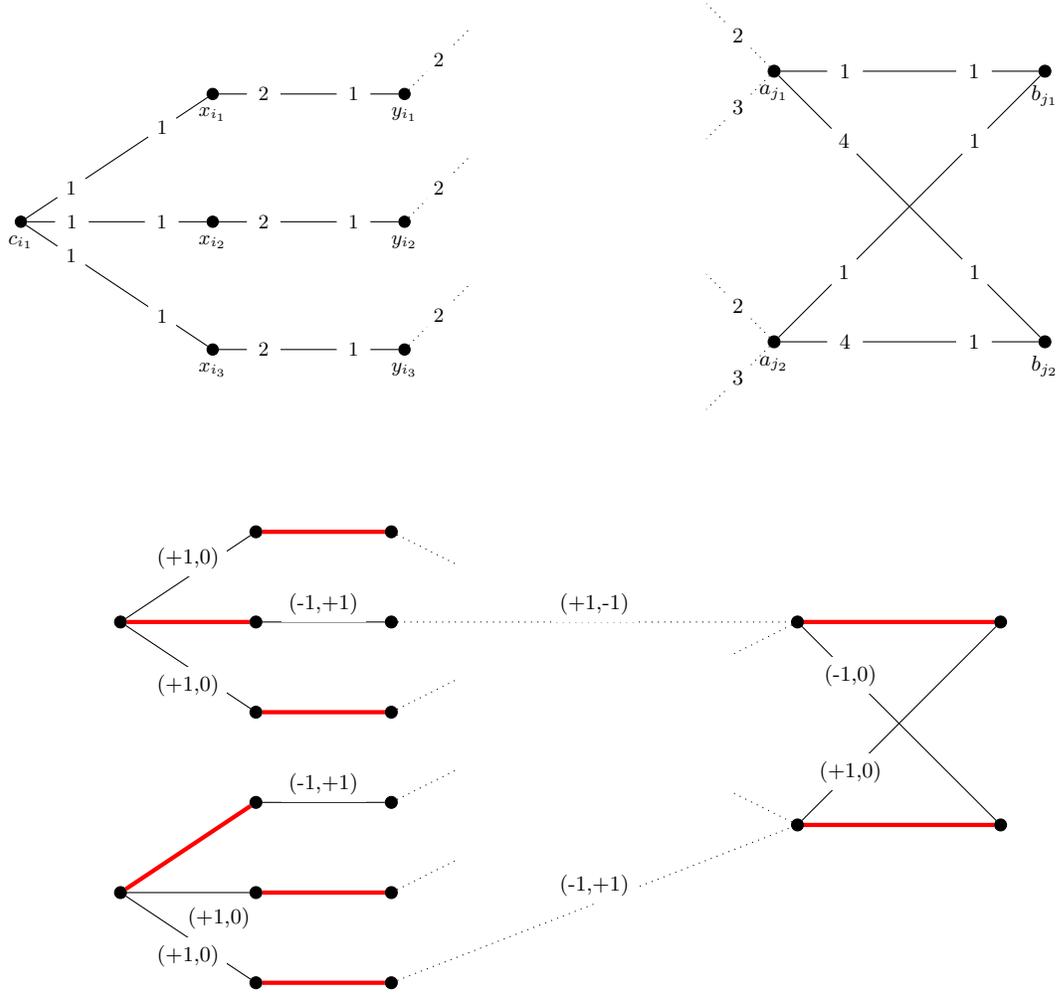
\end{center}

\vspace*{-1cm}
	
\paragraph{Constructing a truth assignment in $\mathcal{I}$, given a popular matching $M$ in~$G$.}
The graph $G$ is as described above. Claim~\ref{clm1:appb} states that any popular matching $M$ in $G$ has a certain structure.
\begin{new-claim}
\label{clm1:appb}
Any popular matching $M$ in $G$ has to obey the following properties.
\begin{itemize}
	\item $M$ avoids all interconnecting edges.
	\item $M$ is one of the two perfect matchings on any variable gadget; i.e.\ for each $j$, the edges of $M$, restricted to the gadget 
corresponding to variable $v_j$, are either (i)~$(a_{j_1},b_{j_1})$ and $(a_{j_2},b_{j_2})$, or (ii)~$(a_{j_1},b_{j_2})$ and $(a_{j_2},b_{j_1})$.
	\item $M$ leaves exactly one vertex per clause $i$ unmatched and this unmatched vertex $y_{i_k}$ is adjacent to an $a_{j_t}$ that is 
matched to~$b_{j_1}$.
\end{itemize}
\end{new-claim}

\begin{proof}
Label each edge $(a,b)$ in $G\setminus M$ by the pair $(\alpha,\beta)$, where $\alpha \in \{\pm 1\}$ is $a$'s vote for $b$ versus $M(a)$ and 
$\beta \in \{\pm 1, 0\}$ is $b$'s vote for $a$ versus~$M(b)$. Our first observation is that every $c_i$, for $1 \leq i  \leq m$, and every 
$b_{j_1}$, for $1 \leq j \leq n$, must be matched in~$M$. That is because these vertices are the top choices for each of their neighbors, 
hence if one of them is left unmatched, then there would be an edge labeled $(+1,+1)$ incident to an unmatched vertex. This contradicts the 
popularity of~$M$.
 
Assume without loss of generality that $\{(c_i, x_{i_3}), (a_{j_1}, b_{j_1})\} \subseteq M$.  Also, the edges $(x_{i_1}, y_{i_1})$ and
$(x_{i_2}, y_{i_2})$ must be in $M$, because they are the top-ranked edges of $y_{i_1}$ and $y_{i_2}$, respectively. We now claim
that $(a_{j_2}, b_{j_2}) \in M$ as well. 

Suppose $(a_{j_2}, b_{j_2}) \notin M$. Since $M$ is a maximal matching, $(a_{j_2},y_{i_k}) \in M$ 
for some~$i_k$. Based on the above described structure of the clause gadgets, the edges $(x_{j_k},c_j), (x_{j_{k+1}}, y_{j_{k+1}})$, and 
$(x_{j_{k+2}}, y_{j_{k+2}})$ are in $M$, where the subscripts are taken modulo~3. 
Consider the following augmenting path $p$ wrt~$M$: 
\[ \rho = b_{j_2} - a_{j_1} - b_{j_1} - a_{j_2} - y_{j_k} - x_{j_k} - c_j - x_{j_{k+1}} - y_{j_{k+1}}.\]
We have $M \oplus p \succ M$, which contradicts the
popularity of~$M$. Thus $(a_{j_2}, b_{j_2}) \in M$. 

An analogous argument proves that if $(a_{j_2}, b_{j_1}) \in M$ for some $j$, then $(a_{j_1}, b_{j_2})$ has to be in~$M$. 
The last observation we make is that if $y_{i_k}$ is unmatched in $M$, then its interconnecting edge leads to an $a_{j_t}$  
that is matched to~$b_{j_1}$. Otherwise $(y_{i_k}, a_{j_t})$ would be labeled $(+1,+1)$ with one vertex unmatched, a contradiction 
again to the popularity of~$M$. This finishes the proof of Claim~\ref{clm1:appb}. \qed
\end{proof}

We assign $\mathsf{true}$ to all variables $v_j$ such that $M \supseteq \{(a_{j_1},b_{j_1}),(a_{j_2},b_{j_2})\}$ and 
$\mathsf{false}$ to all variables $v_j$ such that $M \supseteq \{(a_{j_1},b_{j_2}),(a_{j_2},b_{j_1})\}$.

So the truth value of every variable is uniquely defined and all we need to show is that every clause has a true literal. Assume that in clause $C_i$, all three 
literals are false. The clause gadget has an unmatched vertex $y_{i_k}$ that is adjacent to an~$a_{j_t}$. If the literal is false, 
then $a_{j_t}$ prefers $y_{i_k}$ to $M(a_{j_t}) = b_{j_2}$ and this becomes an edge labeled $(+1,+1)$ with an unmatched end vertex -- this 
contradicts the popularity of~$M$. Hence in every clause, there is at least one true literal and so this is a satisfying assignment.

\paragraph{Constructing a popular matching in $G$, given a truth assignment in~$\mathcal{I}$.}
Here we first construct a matching $M$ in the graph $G$ as described below and then show that it is popular. Initially $M = \emptyset$.
For each $j$, where $1 \le j \le n$, if $v_j = \mathsf{true}$ in the assignment, then add $(a_{j_1}, b_{j_1})$ and $(a_{j_2}, b_{j_2})$ to $M$, 
else add $(a_{j_1}, b_{j_2})$ and $(a_{j_2}, b_{j_1})$ to~$M$. 
For each $i$, where $1 \le i \le m$, in the gadget corresponding to clause $C_i$, any true literal is chosen (say, the $k$-th literal) and 
$y_{i_k}$, representing its appearance, is left unmatched. Moreover, $(x_{i_k}, c_i), (x_{i_{k+1}}, y_{i_{k+1}})$ and $(x_{i_{k+2}}, y_{i_{k+2}})$ (where
the subscripts are taken modulo 3) are added to~$M$. No interconnecting edge appears in~$M$. This finishes the description of~$M$.

\begin{new-claim}
\label{clm2:appb}
The matching $M$ is popular in~$G$.
\end{new-claim}
\begin{proof}
Suppose $M$ is not popular. Then there is another matching $M'$ that is more popular than~$M$. This can only happen if $M \oplus M'$ contains 
a component $\rho$ such that the number of vertices in $\rho$ that prefer $M'$ to $M$ is more than those that prefer $M$ to~$M'$.
To achieve this, the matching $M'$ should contain at least one edge labeled either $(+1,+1)$ or $(+1,0)$, where we use edge labels $(\alpha,\beta)$
as described in the proof of Claim~\ref{clm1:appb}. We now analyze the cases based on the occurrences of such ``positive'' edges.

Since we started with a truth assignment, no interconnecting edge can be labeled $(+1,+1)$. In fact, it is easy to check that no edge here can 
be labeled $(+1,+1)$. We now check for the occurrences of edges labeled $(+1,0)$. These can occur at two places: the edge $(a_{j_t}, b_{j_1})$ 
for any $1 \le j \le n$ and the edge $(x_{i_k},c_i)$ for any $1 \le i \le m$. 

\noindent {\em Case~1.} Suppose $(a_{j_2}, b_{j_1})$ is labeled~$(+1,0)$. This happens if $v_j$ is $\mathsf{true}$ in the truth assignment.
We start the augmenting path $\rho$ at~$(a_{j_2}, b_{j_1})$. Augmenting along the 4-cycle is not sufficient to break popularity, therefore, $a_{j_1}$ must 
be matched along one of its interconnecting edges, say $(a_{j_1},y_{i_k})$. 
\begin{itemize}
\item If $y_{i_k}$ is unmatched, consider the path $\rho = b_{j_2}$-$a_{j_2}$-$b_{j_1}$-$a_{j_1}$-$y_{i_k}$. There are two vertices ($a_{j_1}$ and $b_{j_2}$) 
that prefer $M$ to $M \oplus \rho$ and two vertices ($a_{j_2}$ and $y_{i_k}$) that prefer $M \oplus \rho$ to~$M$.
\item If $y_{i_k}$ is matched, then extend the path $\rho$ till the unmatched vertex of the $i$th variable gadget (call this $y_{i_t}$). The path
$\rho$ is described below:
\[ \rho = b_{j_2} - a_{j_2} - b_{j_1} - a_{j_1} - y_{i_k} - x_{i_k} - c_i - x_{i_t} - y_{i_t}.\]
So 4 vertices, i.e.\ $b_{j_2}$, $a_{j_1}$, $y_{i_k}$, and $x_{i_t}$, prefer
$M$ to $M \oplus \rho$ while 3 vertices, i.e.\ $a_{j_2}$, $x_{i_k}$, and $y_{i_t}$, prefer $M \oplus \rho$ to~$M$.
\end{itemize}


\noindent {\em Case~2.}
Now suppose $(x_{i_k},c_i)$ is labeled $(+1,0)$. Let us assume that this edge is $(x_{i_3},c_i)$ and suppose~$(x_{i_1},c_i) \in M$. 
Consider the alternating path $\rho = y_{i_1}$-$x_{i_1}$-$c_i$-$x_{i_3}$-$y_{i_3}$. In the matching $M \oplus \rho$, the vertices $x_{i_3}$ and
$y_{i_1}$ are better-off while $x_{i_1}$ and $y_{i_3}$ are worse-off, i.e.\ they prefer $M$ to  $M \oplus \rho$. In order to collect one
more vertex that prefers $M \oplus \rho$, let us extend this alternating path $\rho$ to include  $(a_{j_k},y_{i_3})$, the interconnecting edge 
of~$y_{i_3}$. The vertex $y_{i_3}$ still prefers $M$ to $M\oplus\rho$ since $y_{i_3}$ was paired in $M$ to its top-ranked neighbor.

Without loss of generality, let us assume that this interconnecting edge is $(a_{j_2},y_{i_3})$. 
We have two cases here: either $\{(a_{j_1}, b_{j_1}),(a_{j_2}, b_{j_2})\} \subseteq M$ or $\{(a_{j_1}, b_{j_2})$, $(a_{j_2}, b_{j_1})\} \subseteq M$. 
\begin{itemize}
\item In the first case, the path $\rho$ gets extended to $\cdots$-$a_{j_2}$-$b_{j_2}$. So $a_{j_2}$ prefers $M \oplus \rho$ to $M$, however
$b_{j_2}$ is left unmatched in $M\oplus\rho$, so $b_{j_2}$ prefers $M$ to $M\oplus\rho$.
\item In the second case, the path $\rho$ gets extended to $\cdots$-$a_{j_2}$-$b_{j_1}$-$a_{j_1}$-$b_{j_2}$. So $a_{j_1}$ prefers $M \oplus \rho$ 
to $M$, however both $a_{j_2}$ and $b_{j_2}$ prefer $M$ to $M\oplus\rho$.
\end{itemize}
We have analyzed all the cases where edges can labeled $(+1,0)$ and we showed that there is no alternating cycle or path $\rho$ containing an edge
labeled $(+1,0)$ such that $M \oplus \rho \succ M$. Thus $M$ is popular. \qed
\end{proof}
This finishes the proof of Theorem~\ref{thm:nphard} stated in Section~\ref{sec:intro}.

\subsubsection{Conclusions and open problems.} We gave an $O(n^2)$ algorithm for the popular 
matching problem in $G = (A \cup B,E)$ where vertices in $A$ have strict preference lists while 
each vertex in $B$ puts all its neighbors into a single tie and $n = |A \cup B|$.
Our algorithm needs the preference lists of vertices in $A$ to be strict and 
the complexity of the popular matching problem when ties are allowed in the preference lists of 
vertices in $A$ is currently unknown.  

When each $b \in B$ either has a single tie of length at most~3 or a strict preference list (and each $a \in A$ has a strict preference list), we showed that the popular matching problem becomes $\mathsf{NP}$-hard. The complexity of the same problem with ties of length at most~2 instead of~3 is open. 
Another open problem is to extend our algorithm to solve the popular matching problem in the many-to-one 
setting where each post $b$ has a capacity $\capac(b)$
and post $b$ prefers $M_1$ to $M_2$ if $\capac(b) \ge |M_1(b)| > |M_2(b)|$. 

\bibliographystyle{plain}

\begin{thebibliography}{}

\end{thebibliography}


\begin{thebibliography}{10}


\bibitem{AIKM05}
D.J. Abraham, R.W. Irving, T.~Kavitha, and K.~Mehlhorn.
\newblock Popular matchings.
\newblock SIAM Journal on Computing, Vol.~37, No.~4, pages~1030--1045, 2007.

\bibitem{BKS03}
P.~Berman, M.~Karpinski, and Alexander~D. Scott.
\newblock Approximation hardness of short symmetric instances of {MAX-3SAT}.
\newblock Electronic Colloquium on Computational Complexity Report, number 49,
  2003.

\bibitem{BIM09}
P.~Bir\'o, R.W. Irving, and D.F. Manlove.
\newblock Popular matchings in the {M}arriage and {R}oommates problems.
\newblock In {\em Proceedings of the 7th CIAC}: pages 97--108, 2010.

\bibitem{dulmage}
A. Dulmage and N. Mendelsohn. 
Coverings of bipartite graphs. 
In {\em Canadian Journal of Mathematics}, 10: pages 517-534, 1958. 

\bibitem{Gar75}
P.~G\"{a}rdenfors.
\newblock Match making: assignments based on bilateral preferences.
\newblock {\em Behavioural Science}, 20: pages 166--173, 1975.

\bibitem{GGL95}
R.L. Graham, M.~Gr\"{o}tschel, and L.~Lovasz, editors.
\newblock {\em The Handbook of Combinatorics}, chapter 3, Matchings and
  Extensions, by W. R. Pulleyblank, pages 179--232.
\newblock North Holland, 1995.

\bibitem{HK11}
C.-C. Huang and T.~Kavitha.
\newblock Popular matchings in the stable marriage problem.
\newblock In {\em Proceedings of the 38th ICALP}: pages 666--677, 2011.

\bibitem{Kav12}
T. Kavitha.
\newblock Popularity vs Maximum cardinality in the stable marriage setting.
\newblock  In {\em Proceedings of the 23rd SODA}: pages~123-134, 2012.

\bibitem{KMN09}
T.~Kavitha, J.~Mestre, and M.~Nasre.
\newblock Popular mixed matchings.
\newblock In {\em Proceedings of the 36th ICALP}: pages 574--584, 2009.

\bibitem{KN08}
T.~Kavitha and M.~Nasre.
\newblock {\em Note: Optimal popular matchings.}
\newblock Discrete Applied Mathematics, 157(14): pages 3181--3186, 2009.

\bibitem{Mah06}
M.~Mahdian.
\newblock Random popular matchings.
\newblock In {\em Proceedings of the 7th EC}: pages 238--242, 2006.

\bibitem{MS06}
D.F. Manlove and C.T.S. Sng.
\newblock Popular matchings in the {C}apacitated {H}ouse {A}llocation problem.
\newblock In {\em Proceedings of the 14th ESA}: pages 492--503, 2006.

\bibitem{McC06}
R.M. McCutchen.
\newblock The least-unpopularity-factor and least-unpopularity-margin criteria for matching problems with one-sided preferences.
\newblock In {\em Proceedings of the 8th LATIN:} pages 593-604, 2008.

\bibitem{MI09}
E.~McDermid and R.W. Irving.
\newblock Popular matchings: Structure and algorithms.
\newblock In {\em Proceedings of the 15th COCOON}: pages 506--515,  2009.

\bibitem{Mes06}
J.~Mestre.
\newblock Weighted popular matchings.
\newblock In {\em Proceedings of the 33rd ICALP}: pages 715--726, 2006.
\end{thebibliography}

\end{document}